\documentclass[letterpaper,twocolumn,10pt]{article}
\usepackage{usenix2019_v3}
\usepackage{filecontents}
\usepackage{url}
\usepackage{amsthm}
\usepackage{amsfonts}
\usepackage{amssymb}
\usepackage{mathtools}
\usepackage{appendix}
\usepackage{subcaption}
\usepackage{xcolor}
\usepackage{enumitem}
\usepackage{bm}
\usepackage{textcomp}
\usepackage{balance}
\usepackage{makecell}

\title{Fine-grained Poisoning Attack to Local Differential Privacy Protocols \\for Mean and Variance Estimation}
\author{}
\date{}

\DeclarePairedDelimiter\floor{\lfloor}{\rfloor}

\newtheorem{definition}{Definition}
\newtheorem{theorem}{Theorem}
\newtheorem{lemma}{Lemma}

\usepackage[skins]{tcolorbox}

\author{
{\rm Xiaoguang Li\textsuperscript{1,2}\footnotemark[1]}\ \ \
{\rm Ninghui Li\textsuperscript{2}}\ \ \
{\rm Wenhai Sun\textsuperscript{2}}\ \ \
{\rm  Neil Zhenqiang Gong\textsuperscript{3}}\ \ \
{\rm Hui Li\textsuperscript{1}}\ \ \
\\
\\
\textsuperscript{1}\textit{Xidian University},
\textsuperscript{2}\textit{Purdue University},
\textsuperscript{3}\textit{Duke University}
}

\begin{document}

\maketitle
\footnotetext[1]{This work was done when the author was at Purdue university.}

\begin{abstract}
Although local differential privacy (LDP) protects individual users' data from inference by an untrusted data curator, recent studies show that  an attacker can launch a data poisoning attack from the user side to inject carefully-crafted bogus data into the LDP protocols in order to maximally skew the final estimate by the data curator. 

In this work, we further advance this knowledge by proposing a new \textit{fine-grained} attack, which allows the attacker to fine-tune and simultaneously manipulate mean and variance estimations that are popular analytical tasks for many real-world applications. To accomplish this goal, the attack leverages the characteristics of LDP to inject fake data into the output domain of the local LDP instance. We call our attack the output poisoning attack (\textsf{OPA}). We observe a security-privacy consistency where a small privacy loss enhances the security of LDP, which contradicts the known security-privacy trade-off from prior work. We further study the consistency and reveal a more holistic view of the threat landscape of data poisoning attacks on LDP. We comprehensively evaluate our attack against a baseline attack that intuitively provides false input to LDP. The experimental results show that \textsf{OPA} outperforms the baseline on three real-world datasets. We also propose a novel defense method that can recover the result accuracy from polluted data collection and offer insight into the secure LDP design.
\end{abstract}

\vspace{5pt}
\section{Introduction}
Local differential privacy (LDP) \cite{duchi2013local}, a variant of differential privacy \cite{dwork2014algorithmic} in a distributed environment, protects individual user data against an untrusted data collector regardless of the adversary’s background knowledge. Numerous LDP protocols have been proposed for various statistical tasks such as frequency \cite{wang2017locally, wang2019answering, erlingsson2014rappor, wang2020locally, warner1965randomized}, mean/variance \cite{duchi2018minimax, wang2019collecting} and distribution \cite{murakami2019utility, li2020estimating}. LDP has also been integrated into many real-world applications as a \textit{de facto} privacy-preserving data collection tool. For example, Google deployed LDP in Chrome browser to collect users' homepages \cite{erlingsson2014rappor}; Microsoft implemented LDP in Windows 10 to analyze application usage statistics of customers \cite{ding2017collecting}.

Recently, Cao \textit{et al.} \cite{cao2019data} and Cheu \textit{et al.} \cite{cheu2019manipulation} independently studied the security of LDP under \textit{data poisoning attacks} (or called \textit{manipulation attacks} in \cite{cheu2019manipulation}). They found that LDP randomization is very sensitive to data manipulation such that malicious users could send carefully crafted false data to effectively skew the collector's statistical estimate. In particular, the current attacks aims to ``push'' the LDP estimate away from the ground truth as far as possible. In \cite{cheu2019manipulation}, the attacker can inject false data through a group of compromised users to degrade the overall LDP performance. The data poisoning attacks in \cite{cao2019data, wu2021poisoning} promote the items of interest (e.g., in a recommender system) by maximizing the associated statistical estimates, such as frequency and key-value data.

\begin{figure}[!t]
    \centering
    \includegraphics[scale = 0.37]{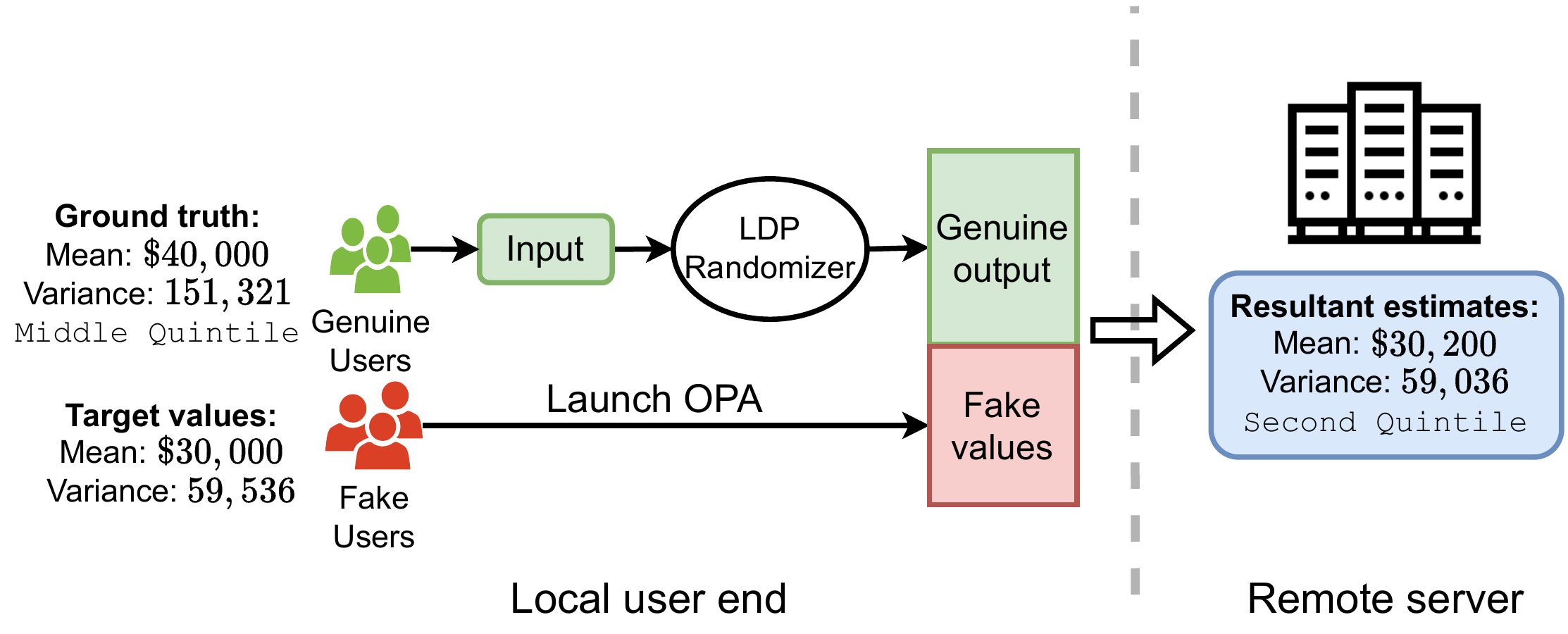}
    \caption{Illustration of our fine-grained data poisoning attacks on LDP-based mean/variance estimation.}
    \label{attack_illustration}
    \vspace{-13pt}
\end{figure}

In this work, we advance the knowledge by considering an attacker who aims to not only manipulate the statistics but also set the estimates to an intended value.  We call it a \textit{fine-grained} data poisoning attack.
We focus on mean and variance estimation because they are crucial to many data analytical applications in practice. For example, a company conducts a market survey to identify the target market segments based on their customers' income average (the mean) and inequality (the variance) \cite{scharfenaker2019labor} as shown in Figure \ref{attack_illustration}. From the survey, the company estimates the mean and variance of the income so as to make informed decisions on the related services. In order to encourage participation, LDP can be adopted to perturb an individual customer's income value before being sent to the company. Meanwhile, a rival company may control a group of fake responders to launch the \textit{fine-grained} data poisoning attack by submitting erroneous information in hopes of bringing the final estimates as close to their intended values as possible. Consequently, the result deviates from reality and leads to a deceptive conclusion, e.g., the customers in the middle-income quintile are mistakenly believed to come from a lower quintile \cite{fontenot2018income}.

We propose an output poisoning attack (\textsf{OPA}) for the \textit{fine-grained} manipulation goal on the local user side against two state-of-the-art LDP protocols for mean and variance, i.e., Stochastic Rounding (SR) \cite{duchi2018minimax} and Piecewise Mechanism (PM) \cite{wang2019collecting}. The attack is illustrated in Figure~\ref{attack_illustration}. Consistent with prior work, we assume that the attacker can control a group of fake users by purchasing accounts from dark markets \cite{cao2019data}, and has access to the LDP implementation details. As a result, the attacker can bypass the LDP perturbation and generate bogus values in the output domain of the local LDP instance, which will be sent to the server for final estimation. To demonstrate the effectiveness of \textsf{OPA}, we compare it with a baseline attack -- we call it an input poisoning attack (\textsf{IPA}), which represents a straightforward data manipulation by providing fake input to the LDP perturbation without leveraging the LDP protocol.

The main challenge for the attacker here is to manipulate two correlated statistical estimates -- mean and variance at the same time through a single LDP query \cite{li2020estimating}. To address this challenge, we formulate the attack as a simultaneous equation-solving problem and coordinate the generation of the poisonous data across the controlled users. To control the LDP estimate at a finer level, the attack also depends on two observations in reality. First, companies and governments, for commercial, public interest or as required by regulations, need to periodically collect user information to learn the \textit{status quo} and then publish the related statistical results \cite{erlingsson2014rappor,microsoft_statistic,jingdong_statistic,2020income}. Second, those historical results regarding the same entity tend to be stable over a short period of time \cite{twitter_statistic2,2020income,fontenot2018income}. As a result, the attacker can leverage the data transparency and the predictable information changes to enable fine-grained data manipulation. Specifically, we assume that the attacker can acquire related background information about genuine users from recent public statistical reports or by compromising a small number of users (see Threat model in Section \ref{threat_model}).

We systematically study the proposed attack both theoretically and empirically. We first analyze the sufficient conditions to launch the attack and further discuss the lower bound on the required number of fake users given the target mean and variance. The results show that \textsf{OPA} needs fewer fake users than the baseline to achieve the same target values.  We are interested in the relationship between various attack parameters and performance, as well as the associated implications. Thus, we also study the MSE between the target value and the final estimate. The results show that \textsf{OPA} has a smaller MSE because direct manipulation of the local LDP output will ignore the effect of perturbation and give the attacker a significant advantage in producing an intended result.

In the literature, a security-privacy trade-off for LDP protocols was revealed: a small $\epsilon$ (strong privacy guarantee) leads to a less secure LDP protocol against prior data poisoning attacks \cite{cheu2019manipulation, cao2019data, wu2021poisoning}. However, we in this work have an opposite observation that weak privacy protection with a large $\epsilon$ is vulnerable to our fine-grained attack. We call this \textit{security-privacy consistency} for LDP protocols. We analyze the two assertions and show that they surprisingly are both valid and that they together provide a holistic understanding of the threat landscape. This conclusion is disturbing since it complicates the already elusive reasoning and selection of the privacy budget in LDP and makes designing a secure LDP more difficult (see Section \ref{security_privacy_consistency}). To mitigate our attack, we also propose a clustering-based method for bogus data tolerance and discuss the relevant defenses in Section \ref{mitigation}. Our main contributions are:
\vspace{-5pt}
\begin{itemize}[leftmargin=*]
\item We are the first to study the \textit{fine-grained data poisoning attack} against the state-of-the-art LDP protocols for mean and variance estimation.

\vspace{-5pt}
\item We propose the \textit{output poisoning attack} to precisely control the statistical estimates to the intended values. By the comparison  with an  LDP-independent baseline attack, i.e., \textit{input poisoning attack}, we show that \textsf{OPA} can achieve better attack performance by taking advantage of LDP. 

\vspace{-5pt}
\item We theoretically analyze the sufficient conditions to launch the proposed attacks, study the attack errors, and discuss the factors that impact the attack effectiveness.

\vspace{-5pt}
\item We discover a fundamental security-privacy consistency in our attacks, which is at odds with the prior finding of the security-privacy trade-off. We provide an in-depth analysis and discussions to reveal the cause of the difference.

\vspace{-5pt}
\item We empirically evaluate our attacks on three real-world datasets. The results show that given the target values, our attacks can effectively manipulate the mean and variance with small errors. We also develop and evaluate a defense method, and provide insights into secure LDP design and other mitigation methods.
\end{itemize}
\section{Background}

\subsection{Local Differential Privacy}
In the local setting of differential privacy, it is assumed that there is no trusted third party. In this paper, we consider there are $n$ users and one remote server. Each user possesses a data value $x \in \mathcal{D}$, and the server wants to estimate the mean and variance of values from all local users. To protect privacy, each user randomly perturbs his/her $x$ using an algorithm $\Psi(x): \mathcal{D} \rightarrow \widehat{\mathcal{D}}$, where $\widehat{\mathcal{D}}$ is the output domain of $\Psi$, and sends $x' = \Psi(x)$ to the server.

\begin{definition}[$\epsilon$-Local Differential Privacy ($\epsilon$-LDP) \cite{duchi2013local}]
    An algorithm $\Psi(\cdot): \mathcal{D} \rightarrow \widehat{\mathcal{D}}$ satisfies $\epsilon$-LDP ($\epsilon > 0$) if and only if for any input $x_1, x_2 \in \mathcal{D}$, the following inequality holds:
    \begin{align*}
        \forall T \in \widehat{\mathcal{D}}, \quad \Pr[\Psi(x_1) = T] \leq e^\epsilon \Pr[\Psi(x_2) = T].
    \end{align*}
\end{definition}
\vspace{-5pt}

Intuitively, an attacker cannot deduce with high confidence whether the input is $x_1$ or $x_2$ given the output of an LDP mechanism. The offered privacy is controlled by $\epsilon$, i.e., small (large) $\epsilon$ results in a strong (weak) privacy guarantee and a low (high) data utility. Since the user only reports the privatized result $\Psi(x)$ instead of the original value $x$, even if the server is malicious, the users' privacy is protected. In our attack, \textit{the attacker can manipulate a group of fake users in order to change the estimates of mean/variance on the server} (See Section \ref{threat_model} for the detailed threat model).

\vspace{-5pt}
\subsection{Mean and Variance Estimation with LDP}
We introduce two widely-used LDP mechanisms for mean and variance estimation, \textit{Stochastic Rounding (SR)} \cite{duchi2018minimax} and \textit{Piecewise Mechanism (PM)} \cite{wang2019collecting}. Note that they were originally developed for mean estimation only and were subsequently adapted to support variance estimation in \cite{li2020estimating}. In this work, we use the adapted version.

\subsubsection{SR mechanism} 
The SR mechanism first uniformly partitions all users into two groups: group $g_1$ reports their original values and group $g_2$ submits their squared original values. All the values must be transformed into $[-1, 1]$ before being used in the LDP. 

\vspace{3pt}
\noindent\textbf{Perturbation.} SR first converts the value into the range $[-1, 1]$. Suppose that the range of the original input values is $[a, b]$. SR calculates transformation coefficients $k_1 = \frac{2}{b - a}$ for $g_1$ and $k_2 = \frac{2}{|b^2 - a^2|}$ for $g_2$, and derives $\tilde{x} = -1 + k_1(x - a)$ for $g_1$ or $\tilde{x} = -1 + k_2(x - a^2)$ for $g_2$. Then SR perturbs values in a discrete output domain with the probability mass function
\begin{align*}
    \Pr[\Psi_{SR(\epsilon)}(\tilde{x}) = x'] = 
    \begin{cases}
        q + \frac{(p - q)(1 - \tilde{x})}{2}, & \text{if } x' = -1 \\
        q + \frac{(p - q)(1 + \tilde{x})}{2}, & \text{if } x' = 1
    \end{cases},
\end{align*}
where $p = \frac{e^\epsilon}{1 + e^\epsilon}$ and $q = 1 - p$. 

\vspace{3pt}
\noindent\textbf{Aggregation.} It has been proven that $\mathbb{E}(\frac{x'}{p - q}) = \tilde{x}$. The server calculates $\Phi_1(x')=(\frac{x'}{p - q} + 1)/k_1 + a$  for $g_1$ and $\Phi_2(x')=(\frac{x'}{p - q} + 1)/k_2 + a^2$ for $g_2$, and estimates their mean. The process provides unbiased estimation of the mean of $x$ and $x^2$, denoted by $\mathbb{E}(x)$ and $\mathbb{E}(x^2)$ respectively. The variance of $x$ is estimated as $\mathbb{E}(x^2) - \mathbb{E}(x)^2$.

\subsubsection{PM mechanism} 

PM also uniformly divides users into groups $g_1$ and $g_2$ in which users report the squared values and original values respectively.

\vspace{3pt}
\noindent\textbf{Perturbation.} In PM, the input domain is $[-1, 1]$ and the output domain is $[-s, s]$, where $s = \frac{e^{\epsilon/2} + 1}{e^{\epsilon/2} - 1}$. Similar to SR, PM first transforms the value into the range $[-1, 1]$ via the same steps in SR. 
Then PM perturbs each value in the  continuous range $[-s, s]$ with the  probability density function as follows
\begin{align*}
    \Pr[\Psi_{PM(\epsilon)}(\tilde{x}) = x'] = 
    \begin{cases}
        \frac{e^{\epsilon/2}(e^{\epsilon/2} - 1)}{2(e^{\epsilon/2}+1)}, & \text{if } x' \in [l(\tilde{x}), r(\tilde{x})] \\
        \frac{e^{\epsilon/2} - 1}{2(e^{\epsilon/2}+e^{\epsilon})}, & \text{otherwise}
    \end{cases},
\end{align*}
where $-s \leq l(\tilde{x}) < r(\tilde{x}) \leq s$, $l(\tilde{x}) = \frac{e^{\epsilon/2} \tilde{x} - 1}{e^{\epsilon/2} - 1}$ and $r(\tilde{x}) = \frac{e^{\epsilon/2} \tilde{x} + 1}{e^{\epsilon/2} - 1}$. 

\vspace{3pt}
\noindent\textbf{Aggregation.} 
It has been proven that $\mathbb{E}(x') = \tilde{x}$ in PM. The server re-converts $x'$ to $\Phi_1(x')=(x' + 1)/k_1 + a$ for $g_1$ and to $\Phi_2(x')=(x' + 1)/k_2 + a^2$ for $g_2$, and then estimates their mean, from which the server can get the unbiased mean estimations $\mathbb{E}(x)$ and $\mathbb{E}(x^2)$. The variance of $x$ is estimated as $\mathbb{E}(x^2) - \mathbb{E}(x)^2$. The following lemma shows the error of the SR and PM mechanisms, which is useful for later analysis of the attack error.
\begin{lemma}[Error of SR and PM mechanisms \cite{wang2019collecting}] \label{mechanism_error}
    Assume there are $n$ users with the values $x_1, ..., x_n$. Let $\mu$ be the mean of those values, and $\hat{\mu}_{SR}$ and $\hat{\mu}_{PM}$ be the mean estimated by the SR and PM respectively. The errors of SR and PM are
    \begin{align*}
        & \mathbb{E}\left[ (\hat{\mu}_{SR} - \mu)^2 \right] = \frac{1}{n^2(p-q)^2} \left( n-(p-q)^2\times \sum_{i=1}^{n} x_i^2 \right) \\
        & \mathbb{E}\left[ (\hat{\mu}_{PM} - \mu)^2 \right] = \frac{e^{\epsilon/2} + 3}{3n(e^{\epsilon/2} - 1)^2} + \frac{\sum_{i=1}^{n} x_i^2}{n^2 (e^{\epsilon/2} - 1)}.
    \end{align*}
\end{lemma}
It is also shown in \cite{wang2019collecting} that the PM mechanism has smaller error than the SR mechanism when $\epsilon$ is large.

\section{Threat Model} \label{threat_model}
In this section, we present our threat model, including the attacker’s capabilities and objectives.

\subsection{Assumption}

Our attacks rely on the following assumptions. First, we assume that the data collector periodically collects user information to derive the intended statistical results. For privacy concerns, LDP may be adopted. This periodical data collection is  important and even mandatory in practice for the update on the \textit{status quo} in order to make informed decisions for relevant activities in the future. For various reasons, such as transparency, research and regulatory compliance \cite{erlingsson2014rappor,microsoft_statistic,jingdong_statistic,2020income,fontenot2018income}, the results will also be made public, thus accessible to the attacker. Second, if the respective data collections are made over a short period of time, the trend of those historical results with respect to the same entity tends to be ``stable'', i.e., their values are close \cite{twitter_statistic2,2020income,fontenot2018income}. Therefore, the attacker can use the statistics from the most recent data report to improve the attack accuracy. Specifically, our attacker needs to estimate the number of authentic users $n$, the sum of the input values of genuine users $S^{(1)}=\sum_{i=1}^{n} x_i$ and the sum of the squared values of genuine users $S^{(2)}=\sum_{i=1}^{n} x_i^2$. Additionally, we assume that the attacker can inject $m$ fake users into the LDP protocol that already contains $n$ genuine users, thus totaling $n+m$ users in the system. This is consistent with prior work showing that an attacker can inject a large number of fake accounts/users into a variety of web services with minimal cost \cite{cao2019data, wu2021poisoning}.
Next, we discuss the estimation of the required information.
\vspace{-.05in}
\begin{itemize}[leftmargin=*]
\item \textbf{Estimating $n$}. Denote $n_e$ as the estimate of $n$. The attacker can deduce $n_e$ from publicly available and reliable sources, e.g., service providers often disclose the number of users under the LDP protection for publicity \cite{erlingsson2014rappor, microsoft_statistic}.
\vspace{-.05in}
\item \textbf{Estimating $S^{(1)}$ and $S^{(2)}$}. Let $S_e^{(1)}$ and $S_e^{(2)}$ be the estimate of $S^{(1)}$ and $S^{(2)}$ respectively. We offer two intuitive estimating methods.

(1) \textit{From public historical data.} This is the most straightforward way. Given the estimated user number $n_e$, the historical mean $\mu_e$ and variance $\sigma_e^2$, the attacker can derive $S_{e}^{(1)} = n_e \times \mu_e$, $S_{e}^{(2)} = (\sigma_e^2 + \mu_e^2) \times n_e$.

\vspace{2pt}
(2) \textit{Compromising a small number of genuine users.} The attacker can compromise $h$ out of $n$ genuine users and obtain their original input values $[c_1, ..., c_h]$. This is reasonable in practice for a small number $h$ and also a prerequisite for prior work \cite{cheu2019manipulation}.
Thus the attacker can estimate $S_e^{(1)} = \frac{n_e}{h} \sum_{i=1}^{h}c_i$, $S_e^{(2)} = \frac{n_e}{h}\sum_{i=1}^{h}c_i^2$.
\end{itemize}
\vspace{-.05in}

Since LDP protocols deploy perturbation on the local user side, we assume that an \textsf{OPA} attacker can gain access to the LDP implementation details, including related parameters and the output domain of the LDP protocol in order to generate and inject bogus data for manipulation. Given a specific LDP protocol, the \textsf{OPA} attacker also knows the group generation strategy for $g_1$ and $g_2$, such that she can 1) determine if the attack would be successful before launching the attack (see in Section~\ref{sec:sufficient_condition}), and 2) craft fake values for each group during the attack (see Section~\ref{sec:OPA_attack}).

For the baseline attack, we only assume that the attacker knows the input domain of the local LDP instance by taking LDP as a black box protocol. 



\subsection{Attack Objectives}
The attacker's goal is to modify the estimated mean $\hat{\mu}_t$ and variance $\hat{\sigma}^2_t$ through LDP to be as close to the target mean $\mu_t$ and variance $\sigma^2_t$ as possible. Meanwhile, the attacker wishes to simultaneously manipulate $\hat{\mu}_t$ and $\hat{\sigma}^2_t$. We adopt the adapted versions of PM and SR mechanisms to privately estimate the mean and variance within a single protocol invocation. Note that our attack objective also implicitly covers the situation of maximizing (minimizing) the mean and variance by setting a significantly large (small) target $\mu_t$ and $\sigma^2_t$. In what follows, we will elaborate on our attacks. Some important notations are summarized in Table \ref{notation}.

\begin{table}[!t]
\centering
\caption{Notations.}
\small
\vspace{-8pt}
\begin{tabular}{|c|l|}
\hline
\textbf{Notation} & \textbf{Description} \\ \hline
$n$        & The number of genuine users          \\ \hline
$n_e$        & The attacker-estimated $n$  \\ \hline
$m$        & The number of fake users           \\ \hline
$\beta$        & The fraction of fake users $\frac{m}{m+n}$          \\ \hline
$g_1$        & The group reporting the original values  $\{x_i\}_{i=1}^{n}$         \\ \hline
$g_2$        & The group reporting the squared values $\{x^2_i\}_{i=1}^{n}$           \\ \hline
$S_{e}^{(1)}$  & The attacker-estimated $\sum_{i=1}^{n} x_i$ \\ \hline 
$S_{e}^{(2)}$ & The attacker-estimated $\sum_{i=1}^{n} x_i^2$ \\ \hline
$\mu_t$  & The attacker's target mean \\ \hline
$\sigma_t^2$  & The attacker's target variance \\ \hline
$k_1$   &   The transformation coefficient for $g_1$ \\ \hline
$k_2$   &   The transformation coefficient for $g_2$ \\ \hline
\end{tabular}
\label{notation}
\vspace{-10pt}
\end{table}

\vspace{-.05in}
\section{Attack Details}
\subsection{Baseline Attack} \label{poisoning_attack_on_input}
We first discuss the baseline attack, \textit{input poisoning attack} (\textsf{IPA}). An \textsf{IPA} attacker does not need to know any details of the underlying LDP protocol. It can submit false data as input to the local LDP instance in order to manipulate the final estimate on the server. Later, we will introduce our attack, \textit{output poisoning attack}, to demonstrate the improved attack performance with knowledge of the LDP protocol compared to the baseline.

Specifically, the goal of \textsf{IPA} is to craft the input values for the controlled fake users in order to alter the mean and variance estimates to be close to the attacker's desired mean $\mu_t$ and variance $\sigma_t^2$.  We generalize the attack for both SR and PM mechanisms. Formally, we denote the original input of $n$ genuine users as $[x_1, ..., x_n]\, (\forall i: x_i \in [a, b])$, and the crafted input of fake users as $[y_1, ..., y_m]\, (\forall i: y_i \in [a, b])$. We formulate \textsf{IPA} as finding $[y_1, ..., y_m]$ such that $\frac{\sum_{i=1}^{n} x_i + \sum_{i=1}^{m} y_i}{n+m}  = \mu_t$ and $\frac{\sum_{i=1}^{n} x_i^2 + \sum_{i=1}^{m} y_i^2}{n+m}  - \mu_t^2 = \sigma_t^2$.

To solve $[y_1, ..., y_m]$, the attacker needs to know $S^{(1)}=\sum_{i=1}^{n} x_i$, $S^{(2)}=\sum_{i=1}^{n} x_i^2$ and $n$, which can be estimated from published information or by compromising a small number of genuine users as described in Section \ref{threat_model}. By substituting $S^{(1)}$, $S^{(2)}$ and $n$ with their estimates $S_{e}^{(1)}$, $S_{e}^{(2)}$ and $n_e$, a set of desired fake values $[y_1, ..., y_m]$ should satisfy
\begin{align}
    & \sum_{i=1}^{m} y_i = (n_e + m) \mu_t - S_{e}^{(1)} \label{input_attack_prac_formula_mean} \\
    & \sum_{i=1}^{m} y_i^2 = (n_e + m) (\sigma_t^2 + \mu_t^2) - S_{e}^{(2)} \label{input_attack_prac_formula_var}.
\end{align}

We transform Equations \eqref{input_attack_prac_formula_mean} and \eqref{input_attack_prac_formula_var} into the following optimization problem and solve it to find valid fake values\footnote{In this work, we use the Adam optimizer in PyTorch framework \cite{pytorch} to solve the problem \eqref{fake_value_optimization} with the learning rate $0.001$ and $10,000$ iterations.}.
\begin{equation}
    \begin{aligned}
        & \min && \left( \sum_{i=1}^{m} y_i^2 - (n_e + m) (\sigma_t^2 + \mu_t^2) - S_{e}^{(2)} \right)^2 \\
        & \mathrm{s.t.} && \sum_{i=1}^{m} y_i = (n_e + m) \mu_t - S_{e}^{(1)}, \forall i: -1 \leq y_i \leq 1 \\
    \end{aligned}
    \label{fake_value_optimization}
\end{equation}

\subsection{Output Poisoning Attack} \label{sec:OPA_attack}
We introduce our output poisoning attack that leverages the LDP implementation details to craft the output of the local perturbation in order to set the final estimates to the target mean $\mu_t$ and variance $\sigma_t^2$.

Let the number of genuine users in $g_1$ and $g_2$ be $n_1$ and $n_2$, and the number of fake users be $m_1$ and $m_2$ respectively. Denote the input of the genuine users in $g_i$ as $x_{1, (i)}, ..., x_{n_i, (i)}$ and the input of the fake users in $g_i$ as $y_{1, (i)}, ..., y_{m_i, (i)}$. Because of the randomness in the LDP local output, the objective of \textsf{OPA} is to produce fake values $\Psi(y_i)~\forall i: 1, ..., m$ such that the expected mean and variance are the attacker-intended $\mu_t$ and $\sigma_t^2$ respectively. 
However, it is difficult to calculate $\mathbb{E}[\hat{\mu}_t^2]$ because $\mathbb{E}[\hat{\mu}_t^2] = Var[\hat{\mu}_t] + \mathbb{E}[\hat{\mu}]^2$ and $Var[\hat{\mu}_t]$ depends on true data. To address this problem, we slack the attack goal by replacing $\mathbb{E}[\hat{\mu}_t^2]$ with $\mu_t^2$. Formally, we intend to achieve the following attack objective in practice.
\begin{align}
    & \mathbb{E} \left[ \frac{2}{n + m} \left( \sum_{i=1}^{n_1} \Phi_1(\Psi(x_{i, (1)})) + \sum_{i=1}^{m_1} \Phi_1(\Psi(y_{i, (1)})) \right) \right] = \mu_t \nonumber \\
    & \Rightarrow \frac{2}{m+n} \left[ \frac{1}{2}S^{(1)} + \sum_{i=1}^{m_1} \Phi_1(\Psi(y_{i, (1)})) \right] = \mu_t \label{output_attack_prac_formula_mean} \\
    & \mathbb{E} \left[ \frac{2}{n + m} \left( \sum_{i=1}^{n_2} \Phi_2(\Psi(x_{i, (2)}^2)) + \sum_{i=1}^{m_2} \Phi_2(\Psi(y_{i, (2)}^2)) \right) \right] - \mathbb{E}[\hat{\mu}_t^2] = \sigma_t^2 \nonumber \\
    &\Rightarrow \frac{2}{m+n} \left[ \frac{1}{2} S^{(2)} + \sum_{i=1}^{m_2} \Phi_2(\Psi(y_{i, (2)}^2)) \right] - \mu_t^2 = \sigma_t^2  \label{output_attack_prac_formula_var}
\end{align}
Since the perturbation $\Psi()$ and aggregation $\Phi()$ are different for SR and PM, the remainder of this subsection will study how to solve Equations \eqref{output_attack_prac_formula_mean} and \eqref{output_attack_prac_formula_var} and generate the fake values accordingly.

\subsubsection{OPA against SR}
By substituting $n$, $S^{(1)}$ and $S^{(2)}$ in Equations \eqref{output_attack_prac_formula_mean} and \eqref{output_attack_prac_formula_var} with their estimates $n_e$, $S_{e}^{(1)}$ and $S_{e}^{(2)}$, we have
\begin{align}
    &\sum_{i=1}^{m_1} \Psi(y_{i, (1)}) = (p-q) \left[ k_1 \left( \frac{n_e + m}{2}\mu_t - \frac{S_{e}^{(1)}}{2} - m_1 a \right) - m_1 \right] \label{output_value_mean} \\
    &\sum_{i=1}^{m_2} \Psi(y_{i, (2)}^2) = (p-q) \times \nonumber \\
    &\left[ k_2 \left( \frac{n_e+m}{2}(\sigma_t^2 + \mu_t^2) - \frac{S_{e}^{(2)}}{2} - m_2 a^2 \right) - m_2 \right] \label{output_value_var}
\end{align}
where $k_1$ and $k_2$ are the transformation coefficients and $a$ is the lower bound of the input range.
In SR, the output is either $-1$ or $1$. Consequently, the attacker can prepare the fake values by determining how many ``$-1$'' and ``$1$'' respectively to be assigned to the fake users. Suppose in group $g_i$ there are $[-1]_{g_i}$ fake users with $-1$ and $[1]_{g_i}$ fake users with $1$. Per Equations \eqref{output_value_mean} and \eqref{output_value_var}, we have
\begin{align*}
\begin{cases}
    [1]_{g_1} + [-1]_{g_1} = m_1 \\
    [1]_{g_1} - [-1]_{g_1} = \sum\limits_{i=1}^{m_1} \Psi(y_{i, (1)})
\end{cases}
\hspace{-12pt}
\begin{cases}
    [1]_{g_2} + [-1]_{g_2} = m_2 \\
    [1]_{g_2} - [-1]_{g_2} = \sum\limits_{i=1}^{m_2} \Psi(y_{i, (2)}).
\end{cases}
\end{align*}
Since there are two unknown variables and two equations, the attacker can solve the above equations to derive the number of $1$ and $-1$ in each group and then randomly assigns them to the fake users in $g_1$ and $g_2$.

\subsubsection{OPA against PM}
In PM, the output value is in the range $[-s, s]$. According to Equations \eqref{output_attack_prac_formula_mean} and \eqref{output_attack_prac_formula_var}, the attacker can calculate the fake values by solving the following equations
\begin{equation*}
\begin{aligned}
    & \sum_{i=1}^{m_1} \Psi(y_{i, (1)}) = k_1 \left( \frac{n+m}{2}\mu_t - \frac{S_{e}^{(1)}}{2} - m_1 a \right) - m_1 \\
    & \sum_{i=1}^{m_2} \Psi(y_{i, (2)}^2) = k_2 \left( \frac{n+m}{2}(\sigma_t^2 + \mu_t) - \frac{S_{e}^{(2)}}{2} - m_2 a^2 \right) - m_2 \\
    & \forall i: \Psi(y_{i, (1)}), \Psi(y_{i, (2)}^2) \in [-s, s].
\end{aligned}
\end{equation*}

An intuitive method to solve the above equations is to divide the right-hand-side by $m_1$ or $m_2$. However, because the fake values generated by this method are equal, the server can easily detect the fake users because it is statistically unlikely that many genuine users will send the same perturbed values. 
To address this problem, the attacker first solves the equations using the method described above, and then randomly perturbs each value while maintaining the sum and keeping the values in $[-s, s]$. Finally, the attacker randomly assigns the values to each fake user in the groups $g_1$ and $g_2$.

\vspace{3pt}
\noindent\textbf{Why is OPA more effective than IPA?}
By accessing the implementation of the underlying LDP protocols, the attacker can generate and inject poisonous data values that are more effective in affecting the server's final estimation. Specifically, the attacker knows how to solve Equations \eqref{output_attack_prac_formula_mean} and \eqref{output_attack_prac_formula_var} by leveraging the knowledge of the LDP perturbation $\Psi()$ and aggregation $\Phi()$. For example, by gaining access to the related parameters, e.g., $p$, $q$, $k_1$, $k_2$, $m_1$ and $m_2$ in $\Psi()$ and $\Phi()$ of SR, the attacker can solve Equations \eqref{output_value_mean} and \eqref{output_value_var}, producing and directly injecting fake values into the output domain of the local LDP instance to launch \textsf{OPA}. As a result, \textsf{OPA} in general will improve the attack performance since the attacker effectively circumvents the LDP perturbation for fake users, thus introducing less noise in the estimation (see the following error analysis).

\section{Theoretical Analysis}
In this section, we theoretically study the sufficient condition to launch the attack and the attack error given the target mean $\mu_t$ and variance $\sigma_t^2$. We assume that the user data in $g_1$ and $g_2$ have been transformed into $[-1, 1]$.

\subsection{Sufficient Condition} \label{sec:sufficient_condition}

\noindent\textbf{Sufficient Condition to Launch IPA.}
The sufficient condition to launch the baseline attack is that Equations \eqref{input_attack_prac_formula_mean} and \eqref{input_attack_prac_formula_var} are solvable so that the attacker can find a set of fake input values of the LDP protocol. Specifically, \textsf{IPA} can be launched if the inequalities hold below.
\begin{align}
    & -m \leq \sum_{i=1}^{m} y_i = (n_e + m)\mu_t - S_{e}^{(1)} \leq m \label{sufficient_condition_1_IPA} \\
    & y^{2(-)} \leq \sum_{i=1}^{m} y_i^2 = (n_e + m) (\sigma_t^2 + \mu_t^2) - S_{e}^{(2)} \leq y^{2(+)} \label{sufficient_condition_2_IPA},
\end{align}
where $y^{2(+)}$ and $y^{2(-)}$ are the maximum and minimum of $\sum_{i=1}^{m} y_i^2$ under the constraint $\sum_{i=1}^{m} y_i = (n_e + m)\mu_t - S_{e}^{(1)}$. Given the input values in the transformed range $[-1, 1]$, \eqref{sufficient_condition_1_IPA} indicates that the sum of all fake values $\sum_{i} y_i$ must reside between $-m$ and $m$ when there are $m$ fake users.It is further required in \eqref{sufficient_condition_2_IPA} that the sum of the squared fake values $\sum_{i} y_i^2$ be in the range $[y^{2(-)},y^{2(+)}]$.

Here we explain how to obtain the above sufficient condition. Since the input value is in the range $[-1, 1]$ and there are $m$ fake users, Equation \eqref{input_attack_prac_formula_mean} is solvable if $-m \leq \sum_{i=1}^{m} y_i = (n_e + m)\mu_t - S_{e}^{(1)} \leq m$ holds. We then need to determine if Equation \eqref{input_attack_prac_formula_var} is solvable under the constraint $\sum_{i=1}^{m} y_i = (n_e + m)\mu_t - S_{e}^{(1)}$. When the range of $\sum_{i=1}^{m} y_i^2$ under this constraint covers the target $\sigma_t^2$, the equation is solvable. To this end, we solve the following optimization problem to find the upper and lower bounds of the term $\sum_{i=1}^{m} y_i^2$. We first study the maximum of $\sum_{i=1}^{m} y_i^2$, i.e., the minimum of $-\sum_{i=1}^{m} y_i^2$.

\begin{theorem} \label{bound_y2}
    Let $A = (n_e+m)\mu_t - S_{e}^{(1)}$, when $\floor{\frac{A + m}{2}}$ fake values are $1$, $m - 1 - \floor{\frac{A + m}{2}}$ fake values are $-1$ and one fake value is $A - \floor{\frac{A + m}{2}} - (m - 1 - \floor{\frac{A + m}{2}})$, $\sum_{i=1}^{m} y_i^2$ achieves the maximum.
\end{theorem}
\begin{proof}
    See Appendix \ref{appendix_solution_max_y2}.
\end{proof}
Similarly, we can determine the lower bound of $\sum_{i=1}^{m} y_i^2$ by changing the objective function from $-\sum_{i=1}^{m} y_i^2$ to $\sum_{i=1}^{m} y_i^2$. We omit the detailed steps here but share the result: when all fake values are $\frac{A}{m}$, $\sum_{i=1}^{m} y_i^2$ achieves the minimum. Given the maximum and minimum of $\sum_{i=1}^{m} y_i^2$ denoted by $y^{2(+)}$ and $y^{2(-)}$ respectively, we can get the above sufficient condition in \eqref{sufficient_condition_1_IPA} and \eqref{sufficient_condition_2_IPA}.


\begin{table*}[!tbp]
\centering
\caption{Comparison of attack error between baseline and our attack against SR and PM. For a concise comparison, we generate some intermediate notations (e.g., $\mathcal{P}, \mathcal{Q}, \mathcal{T}^{IPA}_{\mathrm{SR}}$, etc.) and show their concrete calculation in Table~\ref{tab:intermedias} in Appendix.}
\label{tab:attack_error}
\scalebox{0.9}{
\begin{tabular}{|c|c|c|}
\hline
 & Baseline (IPA) & OPA \\ \hline
$\mathrm{Err}(\hat{\mu}_t)$ in SR & $\mathcal{P} + \frac{2}{(m+n)(p-q)^2} - \mathcal{Q}$ & $\frac{\left( 2n - 2(p-q)^2 S^{(2)} \right)}{(m+n)^2(p-q)^2} + \frac{S^{(2)}}{(m+n)^2} + \mathcal{P}$ \\ \hline
$\mathrm{Err}(\hat{\sigma}^2_t)$ in SR & $\leq \frac{2}{(m+n)(p-q)^2} - \frac{S^{(4)}}{(m+n)^2} + \mathcal{T}^{IPA}_{\mathrm{SR}} + 1$ & $\leq \frac{2n-2(p-q)^2 S^{(4)}}{(m+n)^2(p-q)^2} + \frac{S^{(4)}}{(m+n)^2} + \mathcal{T}^{OPA}_{\mathrm{SR}} + 1$ \\ \hline
$\mathrm{Err}(\hat{\mu}_t)$ in PM & $\frac{2(e^{\epsilon/2} + 3)}{3(n+m)(e^{\epsilon/2} - 1)^2} + \mathcal{P} + \mathcal{Q} + \frac{2 \mathcal{Q}}{(e^{\epsilon/2} - 1)}$ & $\mathcal{P} + \frac{2n(e^{\epsilon/2} + 3)}{3(m+n)^2(e^{\epsilon/2} - 1)^2} + \frac{(1 + e^{\epsilon/2}) S^{(2)}}{(m+n)^2 (e^{\epsilon/2} - 1)}$ \\ \hline
$\mathrm{Err}(\hat{\sigma}^2_t)$ in PM & $\leq \frac{2(e^{\epsilon/2} + 3)}{3(n+m)(e^{\epsilon/2} - 1)^2} + \frac{2 (S^{(4)} + \mathcal{Y}_u^{(4)})}{(n+m)^2 (e^{\epsilon/2} - 1)} + \frac{( S^{(4)} + \mathcal{Y}_u^{(4)} )}{(m+n)^2} + \mathcal{T}^{IPA}_{\mathrm{PM}} + 1$ & $\leq \frac{2n(e^{\epsilon/2} + 3)}{3(m+n)^2(e^{\epsilon/2} - 1)^2} + \frac{(1 + e^{\epsilon/2}) S^{(4)}}{(m+n)^2 (e^{\epsilon/2} - 1)} + \mathcal{T}^{OPA}_{\mathrm{PM}} + 1$ \\ \hline
\end{tabular}
}
\end{table*}

\noindent\textbf{Sufficient Conditions for OPA.}
Now we discuss the sufficient conditions for our attack, which will be analyzed in the context of SR and PM respectively. 

\vspace{2pt}
\noindent\textbullet~ \textit{SR mechanism.}
The sufficient conditions to launch \textsf{OPA} in SR is that Equations \eqref{output_value_mean} and \eqref{output_value_var} are solvable so that the attacker can produce viable output for the local LDP instance in order to manipulate the estimate on the server. In SR, the output is either $-1$ or $1$. Therefore, Equations \eqref{output_value_mean} and \eqref{output_value_var} are solvable if the following hold
\begin{align}
    & -m_1 \leq (p-q) \left( \frac{n_e + m}{2}\mu_t - \frac{S_{e}^{(1)}}{2} \right) \leq m_1 \label{condition_2_SR}, \\
    & -m_2 \leq (p-q) \left( \frac{n_e + m}{2}(\sigma_t^2 + \mu_t^2) - \frac{S_{e}^{(2)}}{2} \right) \leq m_2 \label{condition_1_SR}.
\end{align}

Inequality~\eqref{condition_2_SR} indicates that the sum of fake values $\sum_{i=1}^{m_1}\Psi{(y_{i, (1)})}$ in the output domain should range in $[-m_1, m_1]$ since there are $m_1$ fake users in $g_1$. Inequality~\eqref{condition_1_SR} shows that the range of $\sum_{i=1}^{m_2}\Psi{(y^2_{i, (2)})}$ should be from $-m_2$ to $m_2$ for $m_2$ fake users in $g_2$. By estimating $m_1$ and $m_2$ to be $\frac{m}{2}$, the attacker obtains the sufficient conditions by determining the value of $m$ that satisfies \eqref{condition_2_SR} and \eqref{condition_1_SR}.

\vspace{2pt}
\noindent\textbullet~ \textit{PM mechanism.}
The analysis of PM is similar to that of SR. Equations \eqref{output_value_mean} and \eqref{output_value_var} are solvable if the following inequalities hold. We also estimate $m_1$ and $m_2$ to be $\frac{m}{2}$.
\begin{align}
    & -s\times m_1 \leq \frac{n_e + m}{2}\mu_t - S_{e}^{(1)} \leq s\times m_1 \label{condition_2_PM} \\
    & -s\times m_2 \leq \frac{n_e + m}{2}(\sigma_t^2 + \mu_t^2) - S_{e}^{(2)} \leq s\times m_2 \label{condition_1_PM}
\end{align}
Since the output of PM is within $[-s, s]$, the corresponding upper and lower bounds in \eqref{condition_2_PM} and \eqref{condition_1_PM} are multiplied by $s$.

\vspace{2pt}

\noindent\textbf{Number of Fake Users $m$.} The sufficient condition reveals the relationship between the target values and the needed fake user number $m$, based on which we can further derive the minimum $m$ required to satisfy the sufficient condition given the target mean $\mu_t$ and variance $\sigma_t^2$. Unfortunately, it is challenging to provide the definite mathematical expression of the minimum $m$ because the inequality signs in the condition are uncertain without knowing the values of $n_e$, $S^{(1)}_e$ and $S^{(2)}_e$. On the other hand, since there are only linear and quadratic terms of $m$ in the inequalities, once $n_e$, $S_{e}^{(1)}$, $S_{e}^{(2)}$, $\mu_t$ and $\sigma_t^2$ are given, it is not difficult to get the lower bound on $m$ using the quadratic formula. We empirically study the required minimum number of fake users in Section \ref{section_impact_beta}. The results show that a larger $m$ allows the attacker to set the target values farther from the ground truth. In other words, \textsf{OPA} attacker can control fewer fake users but achieve the same attack efficacy as the baseline. 

\vspace{-10pt}

\subsection{Error Analysis}\label{sec:error_analysis}
In this section, we analyze and compare the attack error of our attack with the baseline attack against SR and PM mechanisms. For ease of analysis, we adopt the widely-used \textit{mean squared error} (MSE) as the error metric.
Specifically, let the estimated mean and variance after the attack be $\hat{\mu}_t$ and $\hat{\sigma}_t^2$. Denote $\mathrm{Err}(\hat{\mu}_t)$ and $\mathrm{Err}(\hat{\sigma}_t^2)$ as the MSE $\mathbb{E}[(\hat{\mu}_t - \mu_t)^2]$ and $\mathbb{E}[(\hat{\sigma}_t^2 - \sigma_t)^2]$ between the target and attack result. Our goal is to study $\mathrm{Err}(\hat{\mu}_t)$ and $\mathrm{Err}(\hat{\sigma}_t^2)$.

We summarize the attack error in Table~\ref{tab:attack_error} and leave the details to Appendices~\ref{appendix_proof_error_input_attack_SR} -- \ref{appendix_proof_error_output_attack_PM}. Note that it is challenging to solve the exact attack error for the variance because the variance estimation depends on the square of a random variable (i.e., the estimated mean). To address this problem, we consider deriving the upper bound of the attack error for variance. 

\vspace{3pt}
\noindent\textbf{SR vs. PM}
We first compare the attack error under different LDP protocols (i.e., SR and PM). We find that the errors are related to the ground truth data due to the terms $S^{(1)}$, $S^{(2)}$ and $S^{(4)}$. By subtracting the error of SR from the error of PM, we find regardless of  \textsf{IPA} or \textsf{OPA},  when $\epsilon$ is small, the error under the SR mechanism is smaller than adopting the PM mechanism given a target mean; when $\epsilon$ is large, the attack against PM performs better because PM introduces less LDP error (see Lemma \ref{mechanism_error}). For a target variance, we cannot draw a similar theoretical conclusion because the analysis only provides the upper bound of the error. But our empirical study in Section~\ref{experiment} shows that as $\epsilon$ grows, the error in PM is smaller than in SR due to less LDP noise introduced.

\vspace{3pt}
\noindent\textbf{OPA vs. Baseline} Now we compare the attack error between \textsf{OPA} and the baseline attack \textsf{IPA}. Theorem \ref{OPA_better} shows that regardless of the underlying LDP protocols, \textsf{OPA} outperforms the baseline with less introduced attack error given a target mean; \textsf{OPA} also has a tighter error upper bound with respect to a target variance.
\begin{theorem} \label{OPA_better}
    The error $\mathrm{Err}(\hat{\mu}_t)$ of \textsf{OPA} is smaller than the error $\mathrm{Err}(\hat{\mu}_t)$ of \textsf{IPA}, and the upper bound of $\mathrm{Err}(\hat{\sigma}^2_t)$ of \textsf{OPA} is smaller than that of \textsf{IPA}.
\end{theorem}
\begin{proof}
    See Appendix \ref{appendix_proof_OPA_better}.
\end{proof}
The intuition is that the submitted fake values in \textsf{IPA} will be perturbed by the LDP, which further contributes to the attack errors. However, \textsf{OPA} is able to bypass the perturbation and directly submit fake values to the server. Therefore, LDP noise does not affect the error calculation. However, Theorem \ref{OPA_better} only states that \textsf{OPA} has a smaller upper bound of error given a target variance. For completeness,  we also empirically study the error. The experimental results show that  \textsf{OPA} outperforms the baseline for both target mean and variance.

\vspace{-10pt}

\section{Consistency of Security and Privacy} \label{security_privacy_consistency}
There is a known security-privacy trade-off in prior research \cite{cao2019data, cheu2019manipulation}, which indicates the incompatible security goal with the privacy requirement of LDP. In other words, prior attacks perform better when $\epsilon$ is set small for higher privacy requirements. However, we do not observe such a trade-off in our proposed data poisoning attack. The security and privacy goals of LDP here are \textit{consistent}, i.e., enhanced privacy also provides improved protection against our data poisoning attack. In this section, we study this consistency for both \textsf{OPA} and \textsf{IPA}, and provide insights into the cause of the difference.

\vspace{-10pt}
\subsection{Security-privacy Consistency in OPA}
We analyze the relationship between the attack performance measured by attack error and the privacy level measured by $\epsilon$ and show the result in Theorem~\ref{consistency_OPA}.

\begin{theorem} \label{consistency_OPA}
    For \textsf{OPA} against SR and PM mechanisms, when the privacy budget $\epsilon$ is larger, the error on mean and the upper bound of the error on variance become smaller.
\end{theorem}
\begin{proof}
    See Appendix~\ref{appendix_proof_consistency_OPA}.
\end{proof}

Theorem \ref{consistency_OPA} only proves that the security-privacy consistency holds for the mean under \textsf{OPA}. The change of the upper bound of the error on variance cannot affirm such consistency result for variance theoretically. Therefore, we also empirically study it and confirm it by our experiments, showing the weakened LDP security as its privacy guarantee deteriorates (see Section \ref{experiment}). 

\subsection{Which is True: Consistency or Trade-off?}
At first glance, the observed security-privacy consistency is at odds with the known result that we have to trade LDP privacy for improved security against unauthorized data injection \cite{cao2019data, wu2021poisoning, cheu2019manipulation}. Through the foregoing analysis and intuitive reasoning, we discover that the two seemingly conflicting findings actually complement each other. They collectively reveal a more holistic view of the threat landscape in the context of data poisoning attacks. We provide the intuition below.

In general, the relationship between LDP security and its privacy depends on the underlying attack objective. In \cite{cheu2019manipulation}, the goal of the attacker is to impair LDP's overall utility. A small $\epsilon$ facilitates the attack by adding more noise to reduce the accuracy of the result. The constructed false values are independent of the privacy budget for the proposed attack in \cite{cao2019data}, which aims to maximize the frequency of target items. A small $\epsilon$ allows the fake users to contribute more to the estimated item frequencies,
resulting in a higher attack gain. In \cite{wu2021poisoning} the security-privacy trade-off remains for the frequency gains of the attack against key-value data \cite{gu2020pckv} since the attack goal is still to maximize the frequency. However, such a trade-off does not necessarily hold when maliciously maximizing the mean. This is because they approximate the mean gain by Taylor expansion in order to perform the attack, which introduces errors into the computation.

Our proposed data poisoning attack has a different goal, i.e., the attacker aims to control the final estimate at a finer level and make the result as close to the target value as possible. Since the attacker can bypass the perturbation and directly inject fake values into the output domain of the local LDP instance, there are two types of errors that impact the result of \textsf{OPA}: the error introduced from the estimation of relevant statistics by the attacker and the error due to the LDP noise from genuine users' input. The former is independent of the LDP privacy implication. For the latter, a small $\epsilon$ incurs large LDP noise such that it is challenging for our attack to precisely manipulate the LDP estimate towards some target value.

\begin{tcolorbox}[enhanced,drop shadow southwest]
\vspace{-.05in}
\textit{The fact that the consistency and trade-off are both valid is disturbing since it complicates the already elusive reasoning and selection of the privacy budget in LDP and makes the design of a secure LDP protocol even more challenging in the presence of different types of data poisoning attacks.}
\vspace{-.05in}
\end{tcolorbox}

We will discuss the mitigation in Section~\ref{mitigation} and the applicability of our attack to other estimations in Section \ref{discussion}.
\vspace{-.05in}

\begin{table}[!t]
\centering
\caption{Dataset information. The numbers in the parentheses are derived from the original user values before being transformed into $[-1, 1]$.}
\scalebox{0.75}{
\begin{tabular}{|c|c|c|c|c|c|}
\hline
Dataset & \#Sample $n$ & $\mu$ & $\sigma^2$ & $S^{(1)}$ & $S^{(2)}$ \\ \hline
\textit{Taxi} \cite{taxi} & 83,130 & \begin{tabular}[c]{@{}c@{}}-0.022\\ (129)\end{tabular} & \begin{tabular}[c]{@{}c@{}}0.34\\ (5,932)\end{tabular} & \begin{tabular}[c]{@{}c@{}}-2,194\\ ($1E7$)\end{tabular} & \begin{tabular}[c]{@{}c@{}}28,362\\ ($1.8E9$)\end{tabular} \\ \hline
\textit{Income} \cite{income} & 2,390,203 & \begin{tabular}[c]{@{}c@{}}-0.93\\ (51,473)\end{tabular} & \begin{tabular}[c]{@{}c@{}}0.007\\ ($4.8E9$)\end{tabular} & \begin{tabular}[c]{@{}c@{}}-2,239,154\\ ($1.2E11$)\end{tabular} & \begin{tabular}[c]{@{}c@{}}2,115,289\\ ($1.8E16$)\end{tabular} \\ \hline
\textit{Retirement} \cite{retirement} & 97,220 & \begin{tabular}[c]{@{}c@{}}-0.87\\ (46,249)\end{tabular} & \begin{tabular}[c]{@{}c@{}}0.025\\ ($3.3E9$)\end{tabular} & \begin{tabular}[c]{@{}c@{}}-85,000\\ ($4.5E9$)\end{tabular} & \begin{tabular}[c]{@{}c@{}}76,752\\ ($5.4E14$)\end{tabular} \\ \hline
\end{tabular}
}
\label{dataset}
\vspace{-15pt}
\end{table}

\begin{table*}[!t]
\centering
\caption{Default parameters. Values in parentheses are derived from the original user values before being transformed into $[-1, 1]$.}
\scalebox{0.9}{
\begin{tabular}{|c|c|c|c|c|c|c|c|c|c|}
\hline
Dataset & $\mu_{t_1}$ & $\mu_{t_2}$ (true $\mu$) & $\mu_{t_3}$ & $\sigma_{t_1}^2$ & $\sigma_{t_2}^2$ (true $\sigma^2$) & $\sigma_{t_3}^2$ & $S_e^{*(1)}$ & $S_e^{*(2)}$ & $n_e^*$ \\ \hline
\textit{Taxi} & \begin{tabular}[c]{@{}c@{}}0.06\\ (140)\end{tabular} & \begin{tabular}[c]{@{}c@{}}-0.022\\ (129)\end{tabular} & \begin{tabular}[c]{@{}c@{}}-0.06\\ (118)\end{tabular} & \begin{tabular}[c]{@{}c@{}}0.33\\ (5,793)\end{tabular} & \begin{tabular}[c]{@{}c@{}}0.34\\ (5,932)\end{tabular} & \begin{tabular}[c]{@{}c@{}}0.4\\ (7,022)\end{tabular} & \begin{tabular}[c]{@{}c@{}} -2,326\\($1E7$)\end{tabular} & \begin{tabular}[c]{@{}c@{}} 29,580\\($1.9E9$) \end{tabular} & 80,000 \\ \hline
\textit{Income} & \begin{tabular}[c]{@{}c@{}}-0.92\\ (65,160)\end{tabular} & \begin{tabular}[c]{@{}c@{}}-0.93\\ (51,473)\end{tabular} & \begin{tabular}[c]{@{}c@{}}-0.94\\ (48,870)\end{tabular} & \begin{tabular}[c]{@{}c@{}}0.005\\ ($3.3E9$)\end{tabular} & \begin{tabular}[c]{@{}c@{}}0.007\\ ($4.8E9$)\end{tabular} & \begin{tabular}[c]{@{}c@{}}0.009\\ ($5.9E9$)\end{tabular} & \begin{tabular}[c]{@{}c@{}}-2,234,086\\($1.2E11$)\end{tabular} & \begin{tabular}[c]{@{}c@{}} 2,108,453\\($2E16$) \end{tabular} & 2,400,000 \\ \hline
\textit{Retirement} & \begin{tabular}[c]{@{}c@{}}-0.86\\ (51,515)\end{tabular} & \begin{tabular}[c]{@{}c@{}}-0.87\\ (46,249)\end{tabular} & \begin{tabular}[c]{@{}c@{}}-0.88\\ (44,156)\end{tabular} & \begin{tabular}[c]{@{}c@{}}0.02\\ ($2.7E9$)\end{tabular} & \begin{tabular}[c]{@{}c@{}}0.025\\ ($3.3E9$)\end{tabular} & \begin{tabular}[c]{@{}c@{}}0.03\\ ($4E9$)\end{tabular} & \begin{tabular}[c]{@{}c@{}} -85,157\\($4.4E9$)\end{tabular} & \begin{tabular}[c]{@{}c@{}} 77,032\\($5.3E14$) \end{tabular} & 100,000 \\ \hline
\end{tabular}
}
\label{parameter_test_mean_var}
\vspace{-5pt}
\end{table*}

\section{Experiments}
\label{experiment}

\subsection{Setup}
\noindent\textbf{Dataset.} 
We used three real-world datasets below to evaluate our attack and baseline attack. They all contain numerical values, which we further converted into $[-1, 1]$. More information about the datasets is summarized in Table~\ref{dataset}.

\noindent\textbullet~ \textit{Taxi}~\cite{taxi}: This dataset comes from 2020 December New York Taxi data, recording the mileage of taxi in a day.
    
\noindent\textbullet~ \textit{Income}~\cite{income}: This dataset contains the income of Americans from the 2019 American Community Survey.
    
\noindent\textbullet~ \textit{Retirement}~\cite{retirement}: This dataset contains the salary paid to retired employees in San Francisco.



\vspace{3pt}
\noindent\textbf{Metric.} 
We repeatedly run our attacks $N=100$ times for each evaluation and record the average. We use MSE to measure the attack performance as this metric is widely used for LDP-related evaluations. Let $\mu_t$ and $\sigma_t^2$ be the target mean and variance, respectively, and the estimated mean and variance in the $i$-th run be $\hat{\mu}_{t_i}$ and $\hat{\sigma}_{t_i}^2$. Formally, we measure
\begin{flalign*}
    \mathrm{MSE}_{\mu} = \frac{1}{N} \sum_{i=1}^{N} \left( \mu_{t} - \hat{\mu}_{t_i} \right)^2,~
    \mathrm{MSE}_{\sigma} = \frac{1}{N} \sum_{i=1}^{N} \left( \sigma_{t}^2 - \hat{\sigma}_{t_i}^2 \right)^2.
\end{flalign*}
Larger MSE implies worse attack performance since the results are farther from the target values. We did not use error bars because the standard deviation is typically very small and barely noticeable in the figures. 


\vspace{3pt}
\noindent\textbf{Parameter Setting.} 
We employ a set of default parameters for the evaluation. As shown in Table~\ref{parameter_test_mean_var}, we have a set of target means $\mu_{t_1}$, $\mu_{t_2}$, $\mu_{t_3}$ and a set of target variances $\sigma_{t_1}^2$, $\sigma_{t_2}^2$, $\sigma_{t_3}^2$ for each dataset, in which $\mu_{t_2}$ and $\sigma_{t_2}^2$ are set to be the true mean and true variance and the rest are randomly produced.  We evaluate two cases: 1) the attacker wants to control mean and variance simultaneously by choosing non-true-value targets; 2) she only attempts to manipulate one while keeping the other as is, i.e. the true value.

We choose the default estimated user number $n_e^*$ based on a common observation that online reports tend to publish round numbers instead of precise values \cite{erlingsson2014rappor}. We also use $\beta = \frac{m}{n+m}$ to denote the ratio of the number of fake users to the number of total users and set the default $\beta = 0.1$. 
$S^{(1)}$ and $S^{(2)}$ can be estimated by either getting access to recent published historical data or compromising a small group of users. Without loss of generality, we randomly select $1,000$ users in each dataset to gain the required information to simulate the public data source and they remain as genuine users in our experiment.

\subsection{Results}
Here we report the experimental results of the sufficient condition evaluation and impact of various factors on the attack.

\begin{figure}[!tb]
    \centering
    \hspace{-0.5cm}
    \begin{subfigure}{0.31\columnwidth}
        \centering
        \includegraphics[scale = 0.17]{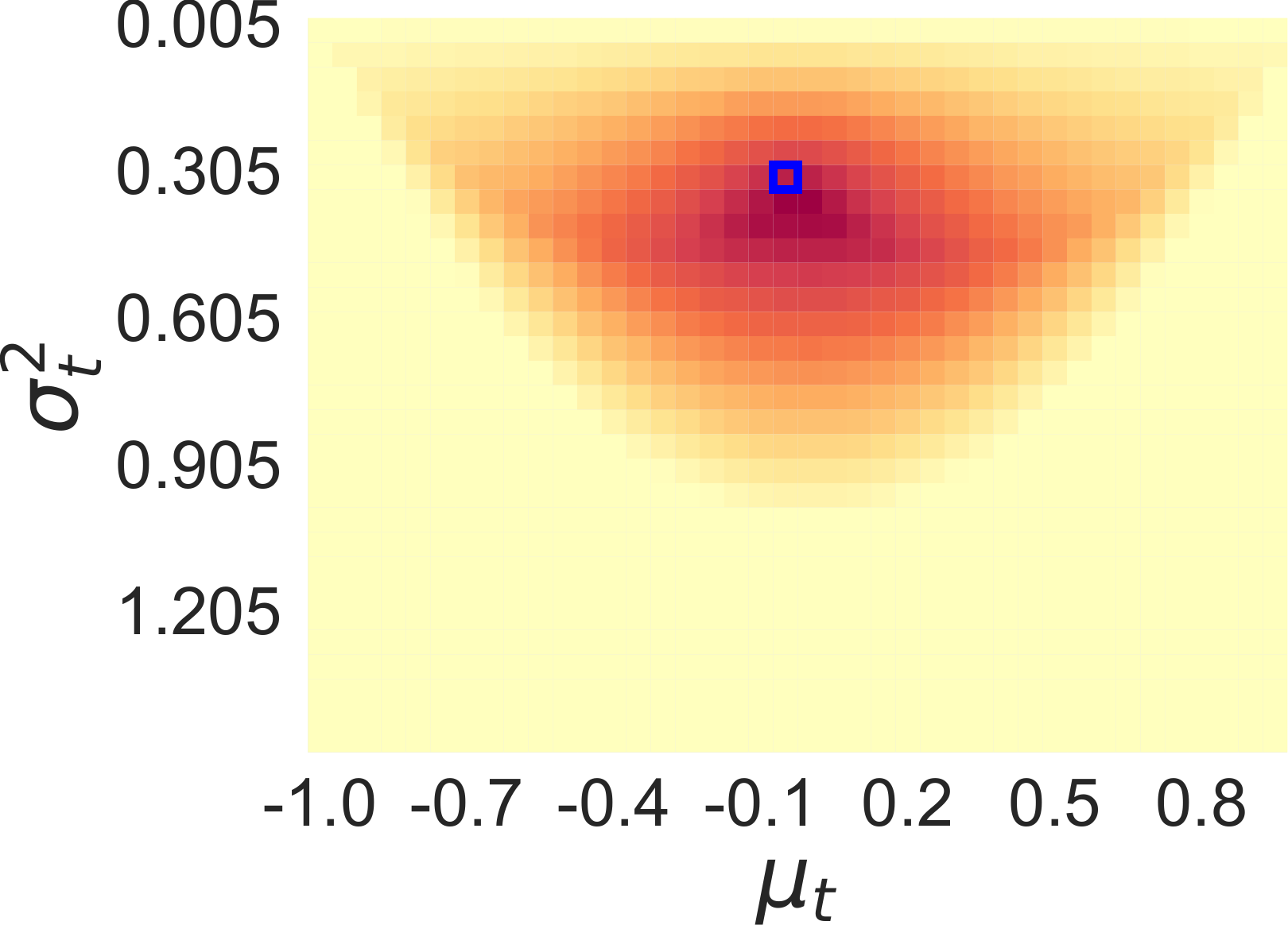}
        \caption{Taxi}
        \label{attack_bound_IPA_taxi}
    \end{subfigure}%
    \hspace{5pt}
    \begin{subfigure}{0.31\columnwidth}
        \centering
        \includegraphics[scale = 0.17]{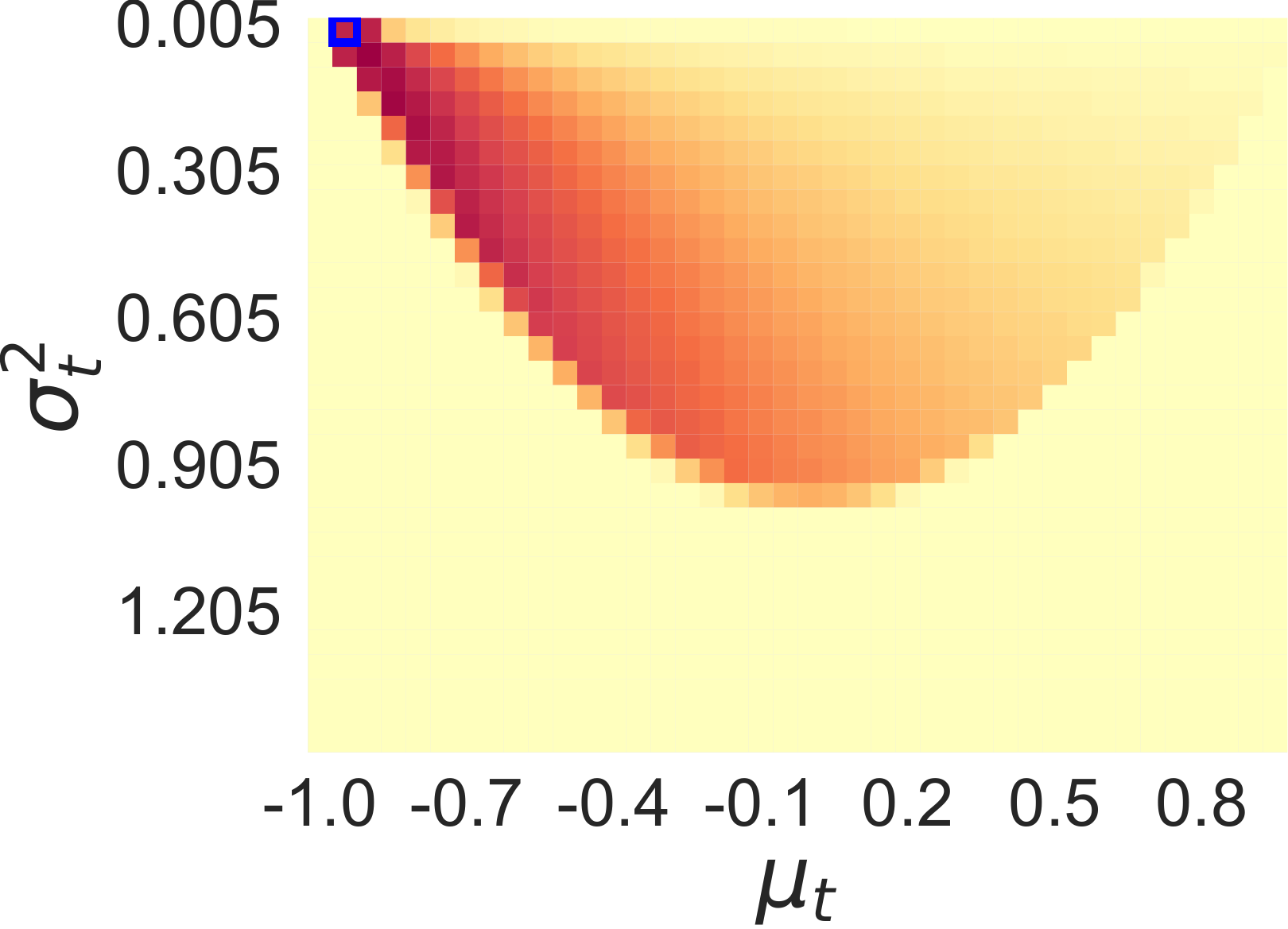}
        \caption{Income}
        \label{attack_bound_IPA_income}
    \end{subfigure}%
    \hspace{5pt}
    \begin{subfigure}{0.31\columnwidth}
        \centering
        \includegraphics[scale = 0.21]{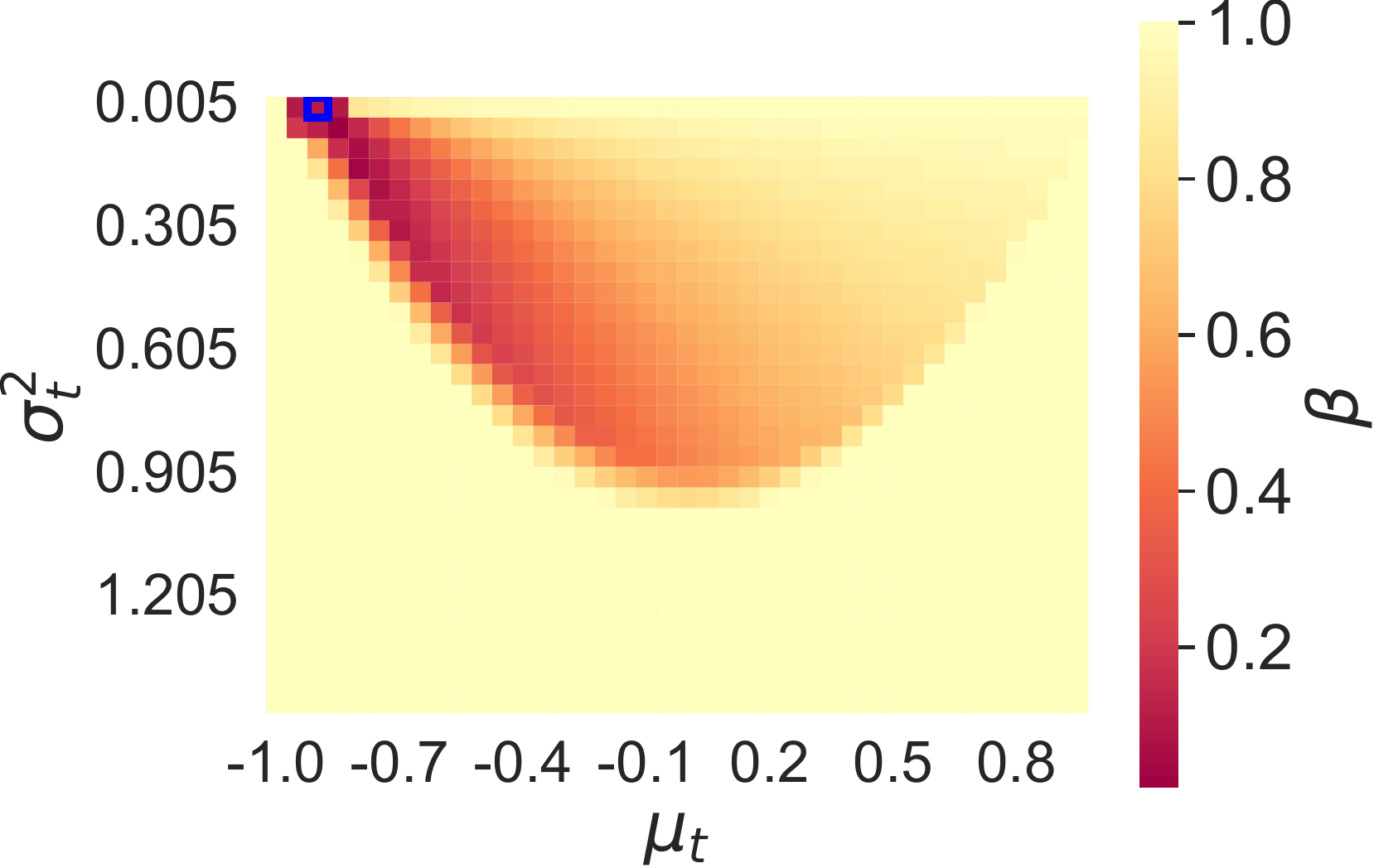}
        \caption{Retirement}
        \label{attack_bound_IPA_retirement}
    \end{subfigure}
    \caption{The minimum number of fake users needed for IPA with $\epsilon = 1$ and varying $\mu_t$ and $\sigma_t^2$ independent of SR and PM.
    }
    \label{attack_bound_IPA}
    \vspace{-10pt}
\end{figure}

\begin{figure}[!tb]
    \centering
    \hspace{-0.5cm}
    \begin{subfigure}{0.31\columnwidth}
        \centering
        \includegraphics[scale = 0.17]{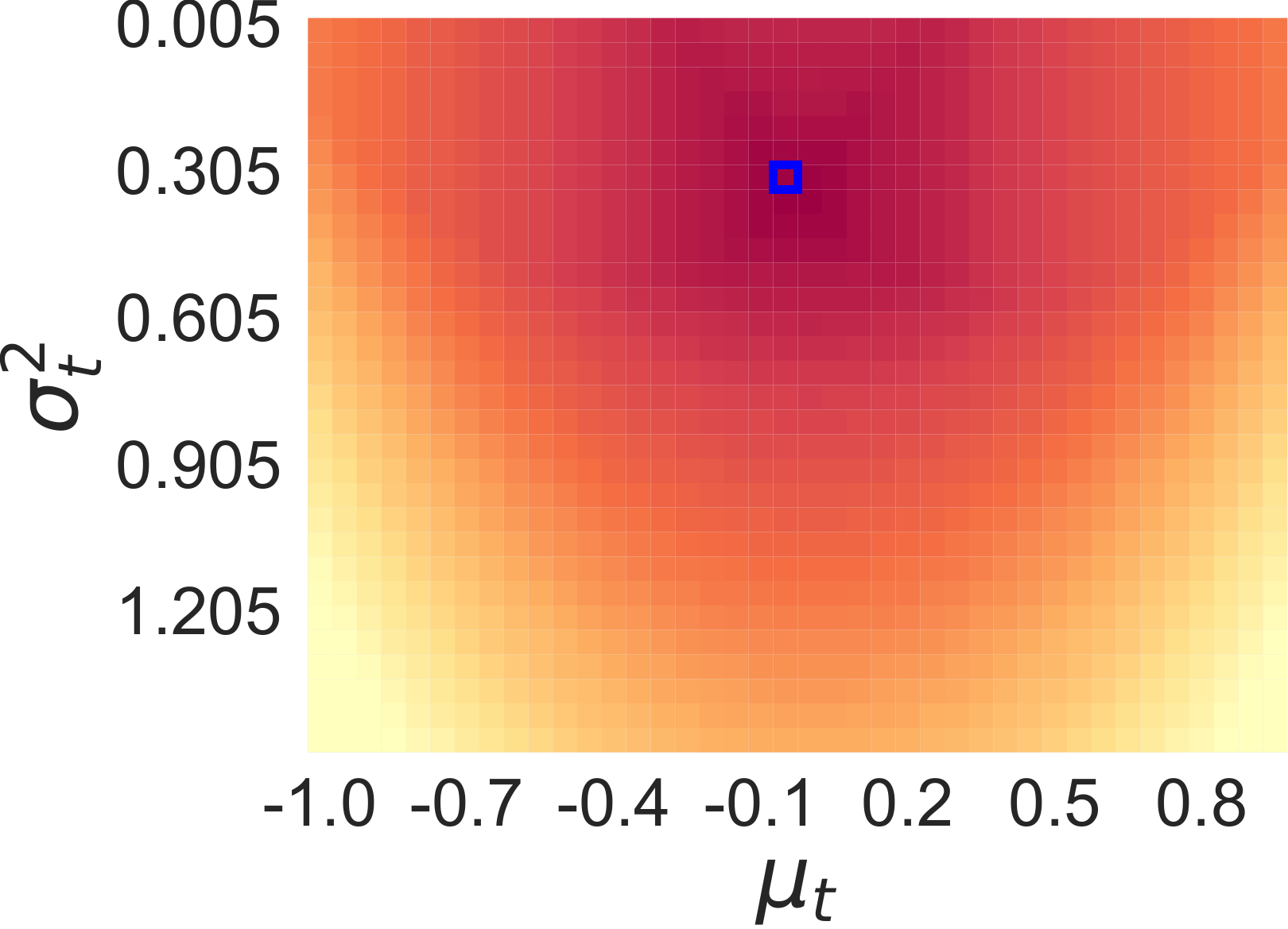}
        \caption{Taxi}
        \label{attack_bound_OPA_SR_taxi}
    \end{subfigure}%
    \hspace{5pt}
    \begin{subfigure}{0.31\columnwidth}
        \centering
        \includegraphics[scale = 0.17]{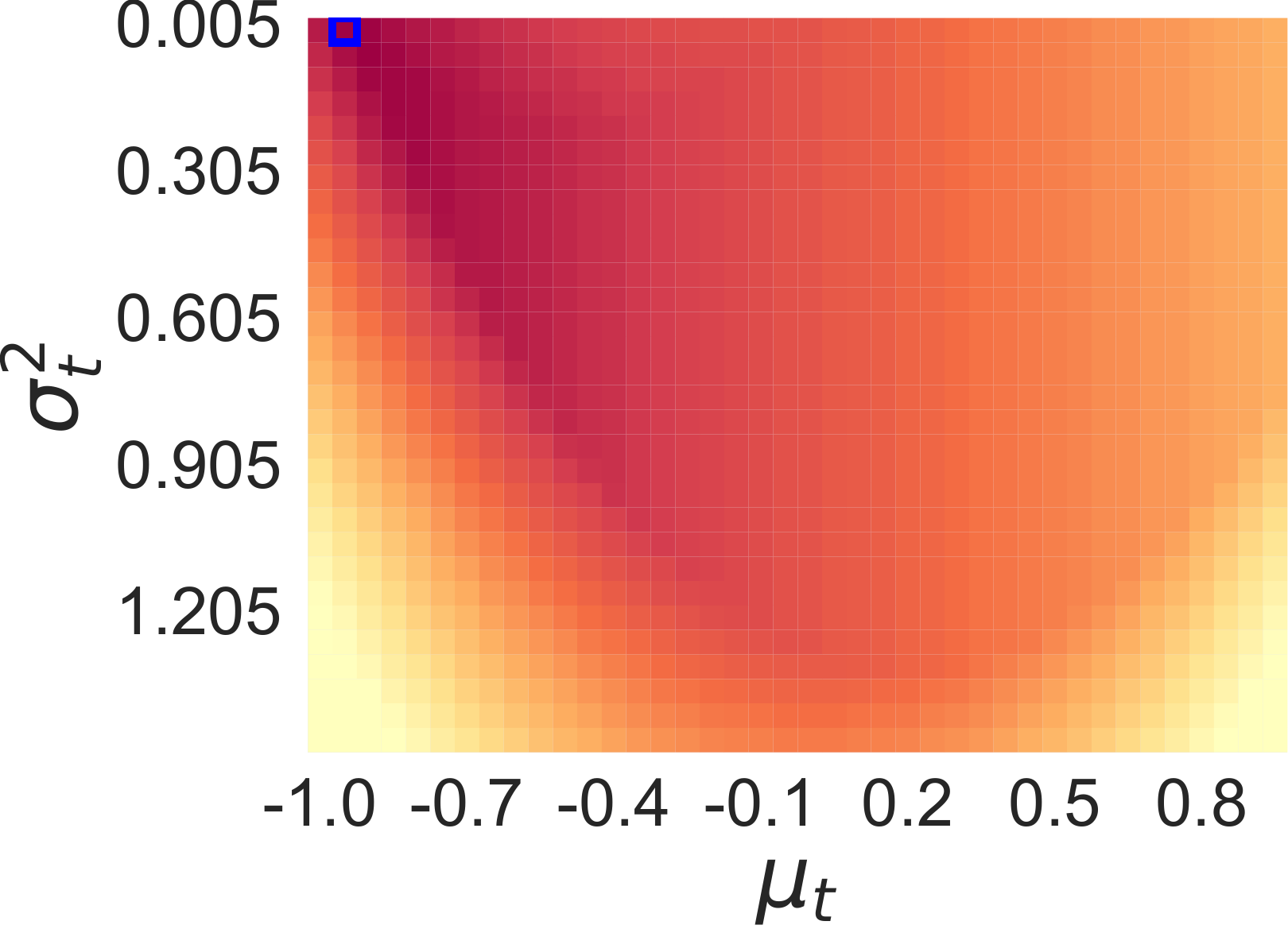}
        \caption{Income}
        \label{attack_bound_OPA_SR_income}
    \end{subfigure}%
    \hspace{5pt}
    \begin{subfigure}{0.31\columnwidth}
        \centering
        \includegraphics[scale = 0.21]{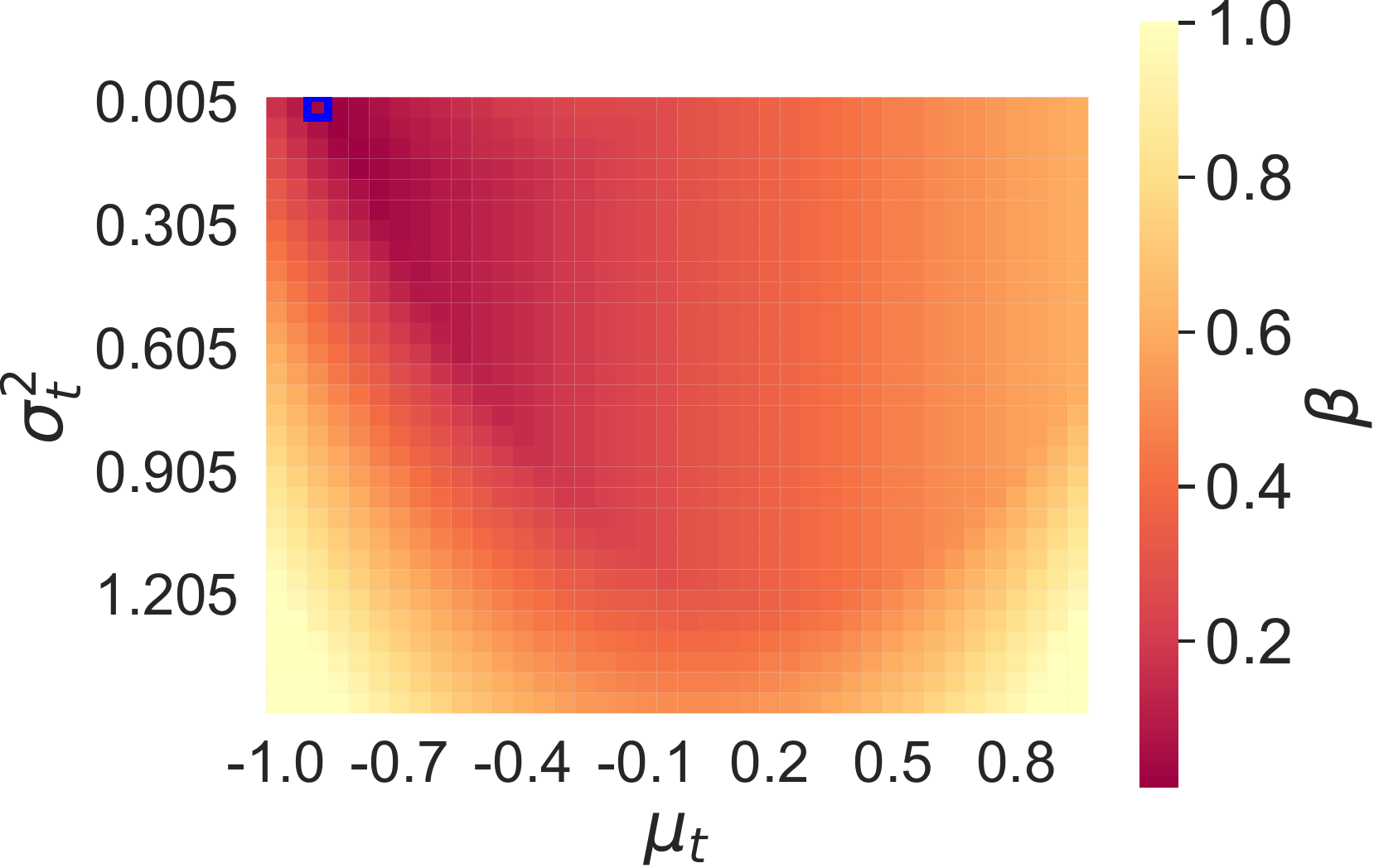}
        \caption{Retirement}
        \label{attack_bound_OPA_SR_retirement}
    \end{subfigure}
    \caption{The minimum number of fake users to launch \textsf{OPA} on SR with $\epsilon = 1$ varying $\mu_t$ and $\sigma_t^2$.
    }
    \label{attack_bound_OPA_SR}
    \vspace{-10pt}
\end{figure}

\begin{figure}[!tb]
    \centering
    \hspace{-0.5cm}
    \begin{subfigure}{0.31\columnwidth}
        \centering
        \includegraphics[scale = 0.17]{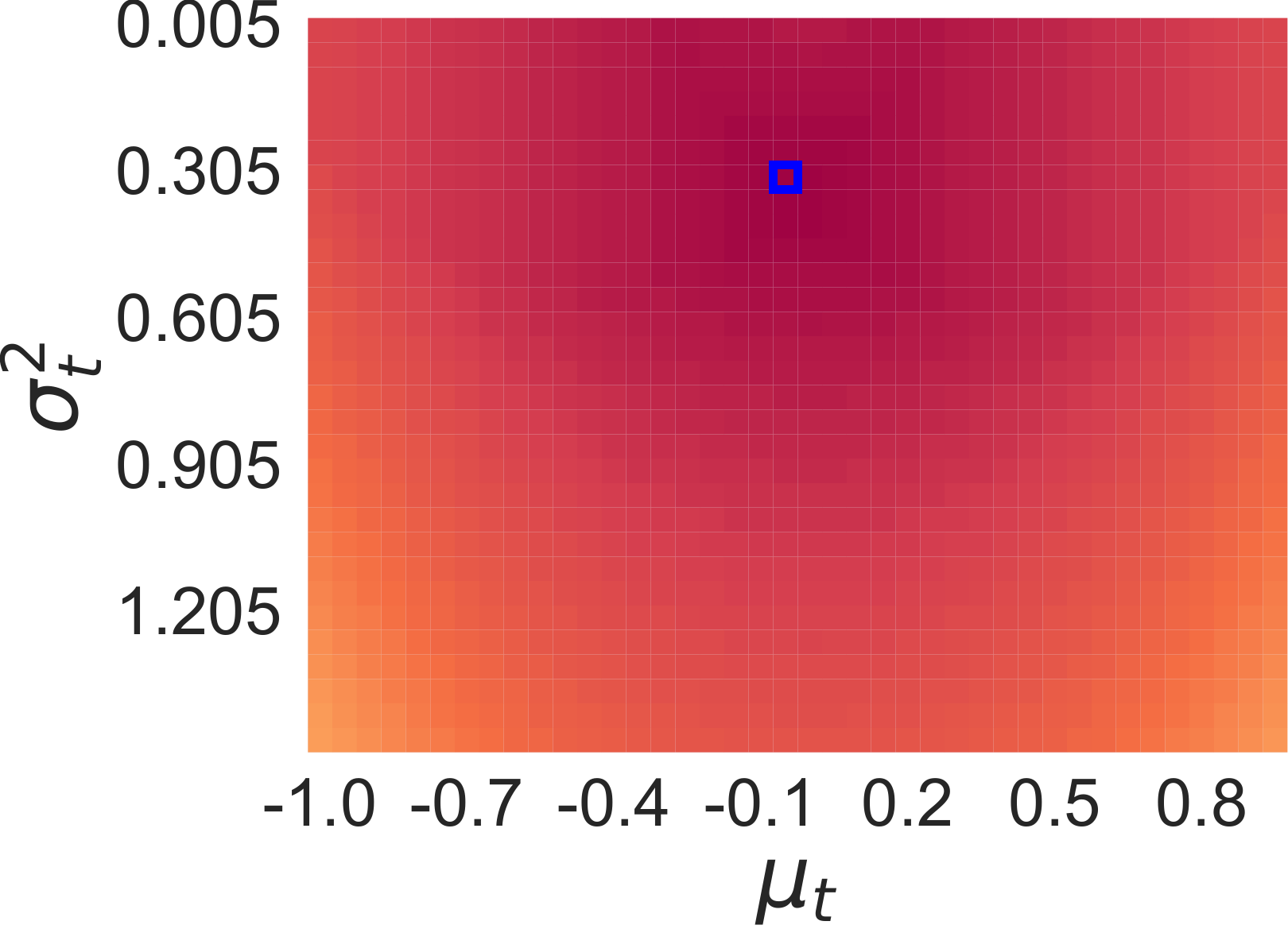}
        \caption{Taxi}
        \label{attack_bound_OPA_PM_taxi}
    \end{subfigure}%
    \hspace{5pt}
    \begin{subfigure}{0.31\columnwidth}
        \centering
        \includegraphics[scale = 0.17]{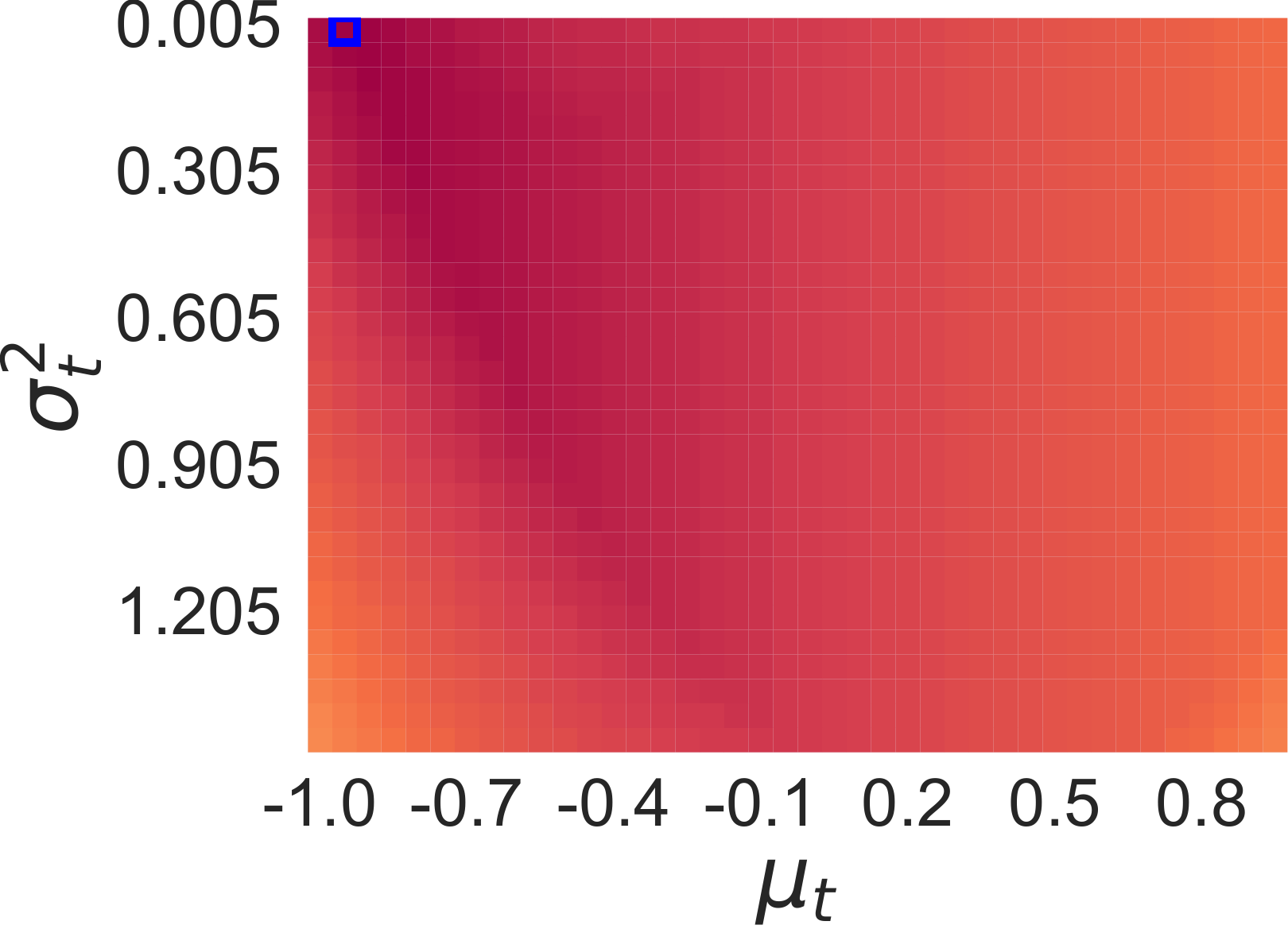}
        \caption{Income}
        \label{attack_bound_OPA_PM_income}
    \end{subfigure}%
    \hspace{5pt}
    \begin{subfigure}{0.31\columnwidth}
        \centering
        \includegraphics[scale = 0.21]{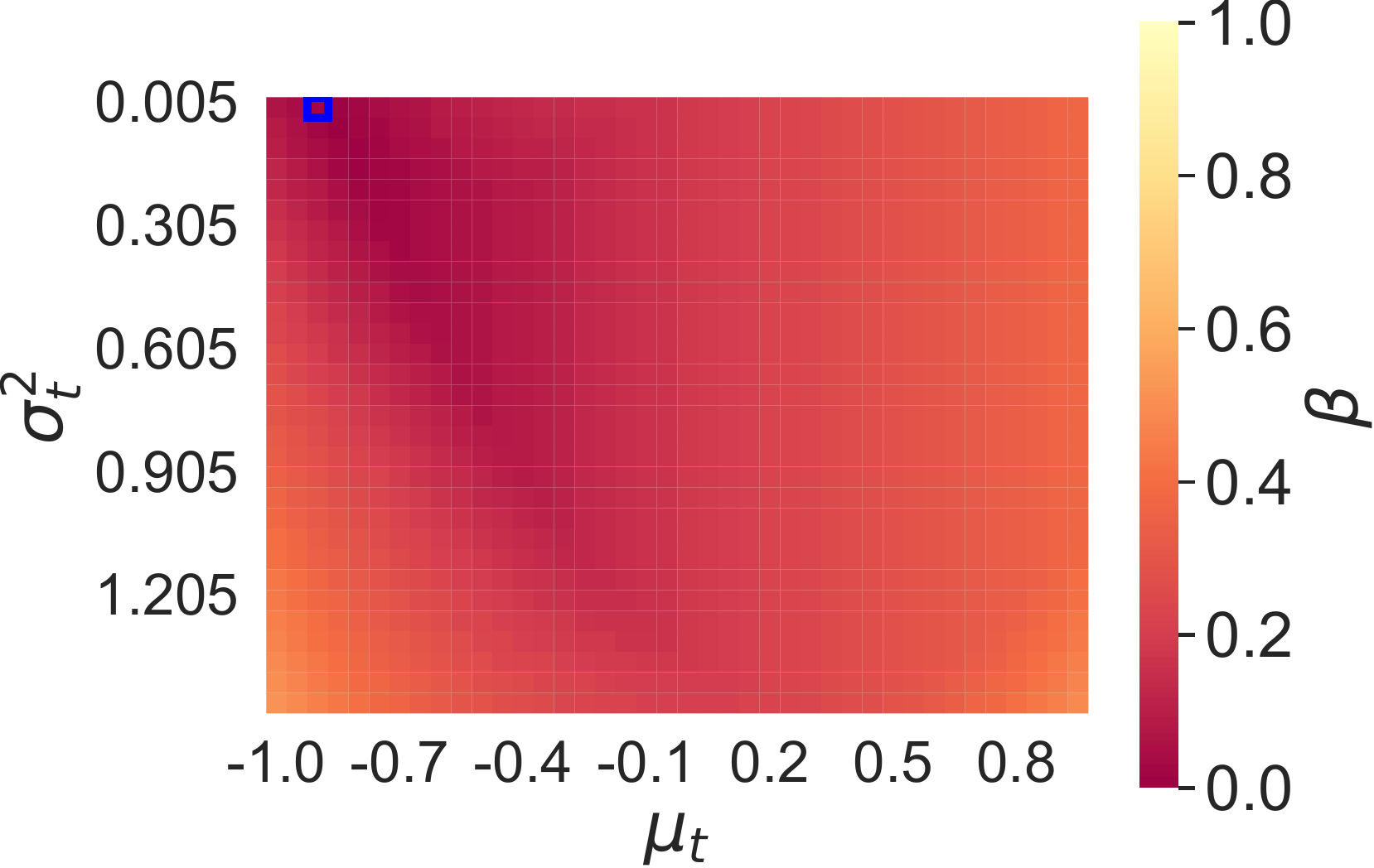}
        \caption{Retirement}
        \label{attack_bound_OPA_PM_retirement}
    \end{subfigure}
    \caption{The minimum number of fake users needed for \textsf{OPA} on PM with $\epsilon = 1$ and varying $\mu_t$ and $\sigma_t^2$. 
    }
    \label{attack_bound_OPA_PM}
    \vspace{-10pt}
\end{figure}

\subsubsection{Sufficient Condition for Achieving Target Values} \label{sec:experiment_sufficient_condition}
We first study the sufficient condition for \textsf{IPA} and \textsf{OPA} to achieve the $\mu_t$ and $\sigma_t^2$ on SR and PM. Specifically, given the attacker-estimated $S^{(1)}_e$, $S^{(2)}_e$ and $n_e$, we study the relationship between $\mu_t$ and $\sigma_t^2$ and the minimum number of fake users (measured by $\beta = \frac{m}{n+m}$) required to launch the attack. The results are shown in Figure~\ref{attack_bound_IPA}, \ref{attack_bound_OPA_SR} and Figure~\ref{attack_bound_OPA_PM}. Since the baseline \textsf{IPA} is independent of LDP, the required minimum number of fake users is the same for both SR and PM. Thus, we use one plot for each dataset regardless of the underlying LDP protocols in Figure \ref{attack_bound_IPA}. \textsf{OPA} leverages the underlying LDP mechanisms, i.e. SR in Figure \ref{attack_bound_OPA_SR} and PM in Figure \ref{attack_bound_OPA_PM}. We highlight the true mean and variance in the encoded range $[-1, 1]$ with a blue rectangle, and use the darker color to indicate that fewer fake users are needed for a successful attack. The light-color parts ($\beta = 1$) represent the extreme cases where the attack is infeasible given a pair of $\mu_t$ and $\sigma_t^2$. We can derive the following key observations: 
\begin{itemize}[leftmargin=*]
\item In general, more fake users allows the attacker in both baseline and our attack to set a target value that deviates farther from the ground true. $\beta$ keeps small when $\mu_t$ and $\sigma_t^2$ grow simultaneously. This is because when $\mu_t$ and $\sigma_t^2$ grow, both constraint terms for $\sum_i^m y_i$ and $\sum_i^m y_i^2$ (in Equations \eqref{input_attack_prac_formula_mean} and \eqref{input_attack_prac_formula_var}) increase together and thus a small number of large fake values can satisfy the constraints. 

\item \textsf{OPA} can reach more darker regions compared to the baseline attack, which means our attacker has a wider selection of target mean and variance pairs and can successfully accomplish the attack given the selected targets. This is because \textsf{OPA} attacker can inject fake values into the LDP output domain. Compared with \textsf{IPA}, the constraints of the fake values are relaxed by the factors in the LDP aggregation (Section~\ref{sec:sufficient_condition}), thus fewer fake users needed for \textsf{OPA}.

\item For \textsf{OPA}, the accessible region of target values in SR is smaller than in PM. This is because the factors in SR aggregation are smaller than those in PM, leading to a set of tighter constraints for fake values and thus the target values spreading over a smaller region.
\end{itemize}

\begin{figure}[!tb]
    \centering
    \includegraphics[scale = 0.16]{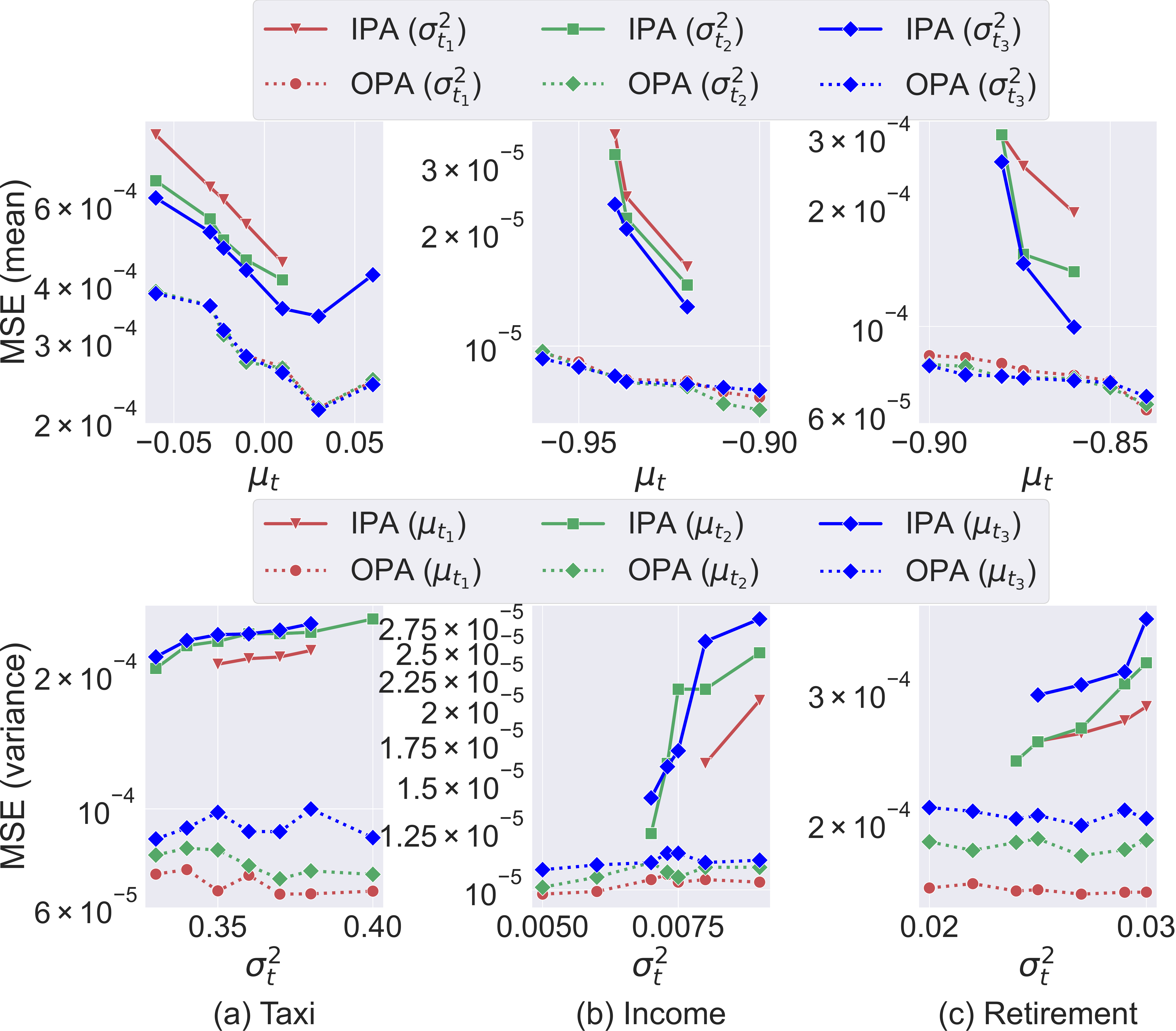}
    \caption{Attack error in SR, varying $\mu_t$ and $\sigma_t^2$. Set $\epsilon = 1$, $m = 0.1n$, $S_{e}^{(1)} = S_{e}^{*(1)}$ and $S_{e}^{(2)} = S_{e}^{*(2)}$, $n_e = n_e^*$.}
    \label{MSE_target_mean_var_SR}
    \vspace{-13pt}
\end{figure}

\vspace{-15pt}
\subsubsection{Impact of Target Values}
Figure \ref{MSE_target_mean_var_SR} and \ref{MSE_target_mean_var_PM} depict the attack performance against SR and PM with varying target mean and variance. We observe that 
\begin{itemize}[leftmargin=*]
    \item \textsf{OPA} outperforms  \textsf{IPA} when attacking the mean, since \textsf{OPA} circumvents LDP perturbation on the fake user side, resulting in less attack error. For example, on \textit{Taxi}, \textsf{OPA} reduces MSE of \textsf{IPA} by about 30\% on SR and PM respectively.
    
    \item \textsf{OPA} also outperforms \textsf{IPA} when attacking the variance with all three datasets. As the target variance grows, the MSE of the baseline \textsf{IPA} against both SR and PM increases. This is because when the target increases, the bias in the attack error grows in the SR mechanism, and both variance and bias increase in PM. From Figure \ref{MSE_target_mean_var_SR}, we observe a much reduced MSE with \textsf{OPA}, e.g., at most 50\% error reduction compared to the baseline against SR by controlling the same number of fake users.
    
    \item For some target mean/variance, the default $\beta$ is inadequate to launch the baseline attack. Thus, no corresponding MSE is recorded within the given range in the figures. The result again exhibits the advantage of \textsf{OPA} over the baseline when the attacker can only control a limited number of users. 
    

    \item Given a target $\mu_t$, the MSE of \textsf{IPA} becomes smaller (larger) on SR (on PM) as $\sigma_t^2$ grows. This is because a larger $\sigma_t^2$ provides a larger $\sum_{i} y_i^2$, thus leading to small error in SR and larger error in PM. However, \textsf{OPA} errors are not affected by $\sum_{i} y_i^2$ and are close  under different $\sigma_t^2$.
    
   \item The MSE of \textsf{IPA}$(\sigma_{t_3}^2)$ and \textsf{OPA}$(\sigma_{t_3}^2)$ on \textit{Taxi} reduces first then increases as  $\mu_t$ grows. This is because both MSE of \textsf{IPA} and \textsf{OPA} are a quadratic function of $\mu_t$ (see error analysis in Section \ref{sec:error_analysis}). On \textit{Income} and \textit{Retirement}, the range of $\mu_t$ only covers the reduction part of quadratic MSE. Figure \ref{MSE_target_mean_var_SR} and \ref{MSE_target_mean_var_PM} thus do not show a parabola shape.
\end{itemize}
\vspace{-5pt}

To have a more intuitive comparison with the target, we observe that the \textsf{OPA} result deviates from the target mean by as small as $0.22\%$  in the figures versus $0.58\%$ with \textsf{IPA} and a maximum error reduction from $17\%$ with \textsf{IPA} to only $3.7\%$ with \textsf{OPA} given a target variance.

\begin{figure}[!tb]
    \centering
    \includegraphics[scale = 0.16]{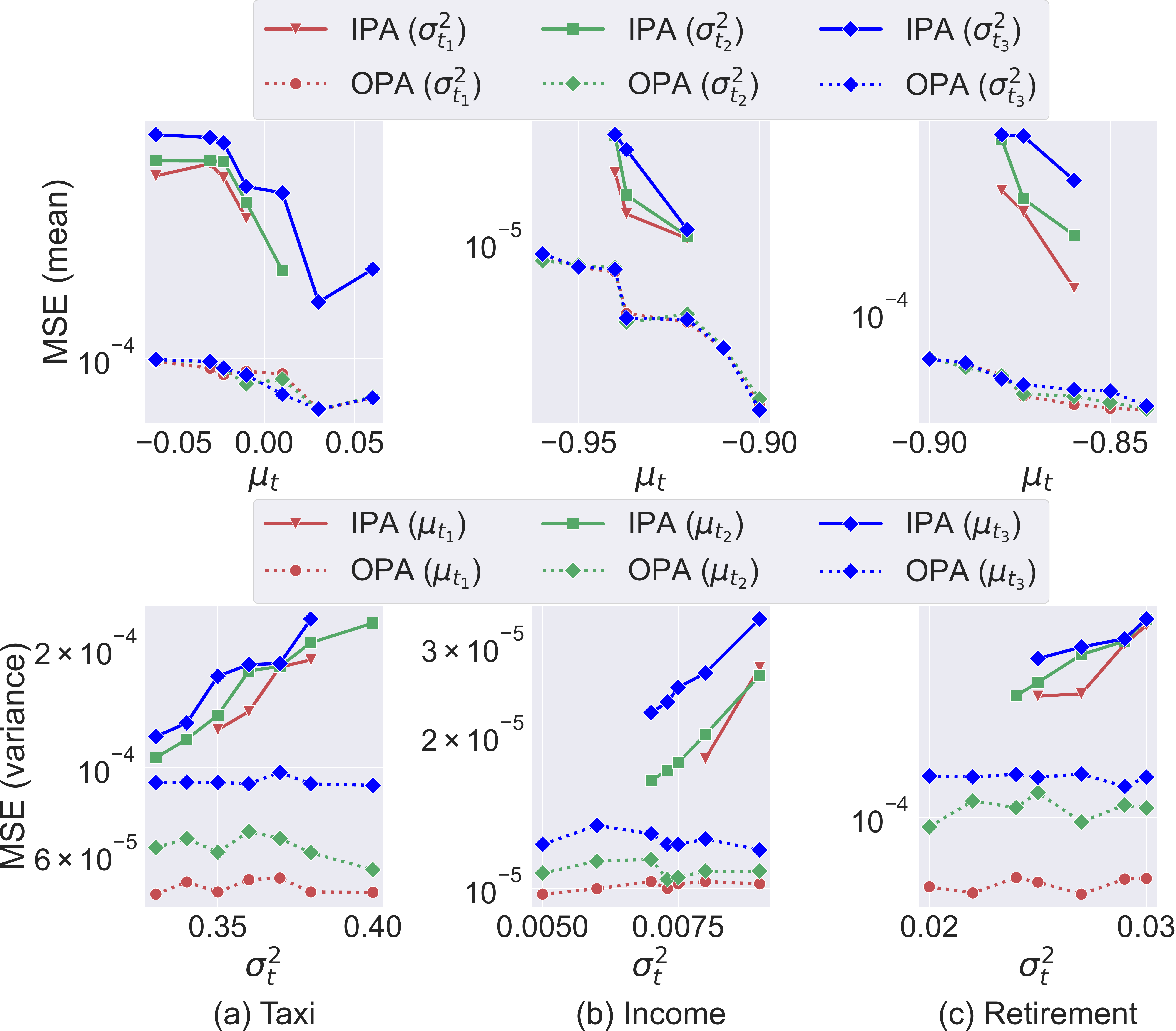}
    \caption{Attack error in PM, varying $\mu_t$ and $\sigma_t^2$. Set $\epsilon = 1$, $m = 0.1n$, $S_{e}^{(1)} = S_{e}^{*(1)}$ and $S_{e}^{(2)} = S_{e}^{*(2)}$, $n_e = n_e^*$.}
    \label{MSE_target_mean_var_PM}
     \vspace{-15pt}
\end{figure}



\vspace{-5pt}
\subsubsection{Impact of Attacker's Estimation}\label{section_impact_Xe_XSe}

We study how the relevant estimation of the attacker affects the results next. In general, more accurate estimation pushes the attack result further to the target.

\noindent\textbf{Impact of $S_{e}^{(1)}$ and $S_{e}^{(2)}$.} 
We discuss the impacts of $S_{e}^{(1)}$ (Figure~\ref{MSE_target_mean_var_Xe_composition}) and $S_{e}^{(2)}$ (Figure~\ref{MSE_target_mean_var_XSe_composition}) on attack performance. Note that the intervals in the $x$-axis of both figures are different across the datasets because the respective $S^{(1)}$ and $S^{(2)}$ are distinct. Below are some key observations.


\begin{figure}[!tb]
    \centering
    \includegraphics[scale = 0.16]{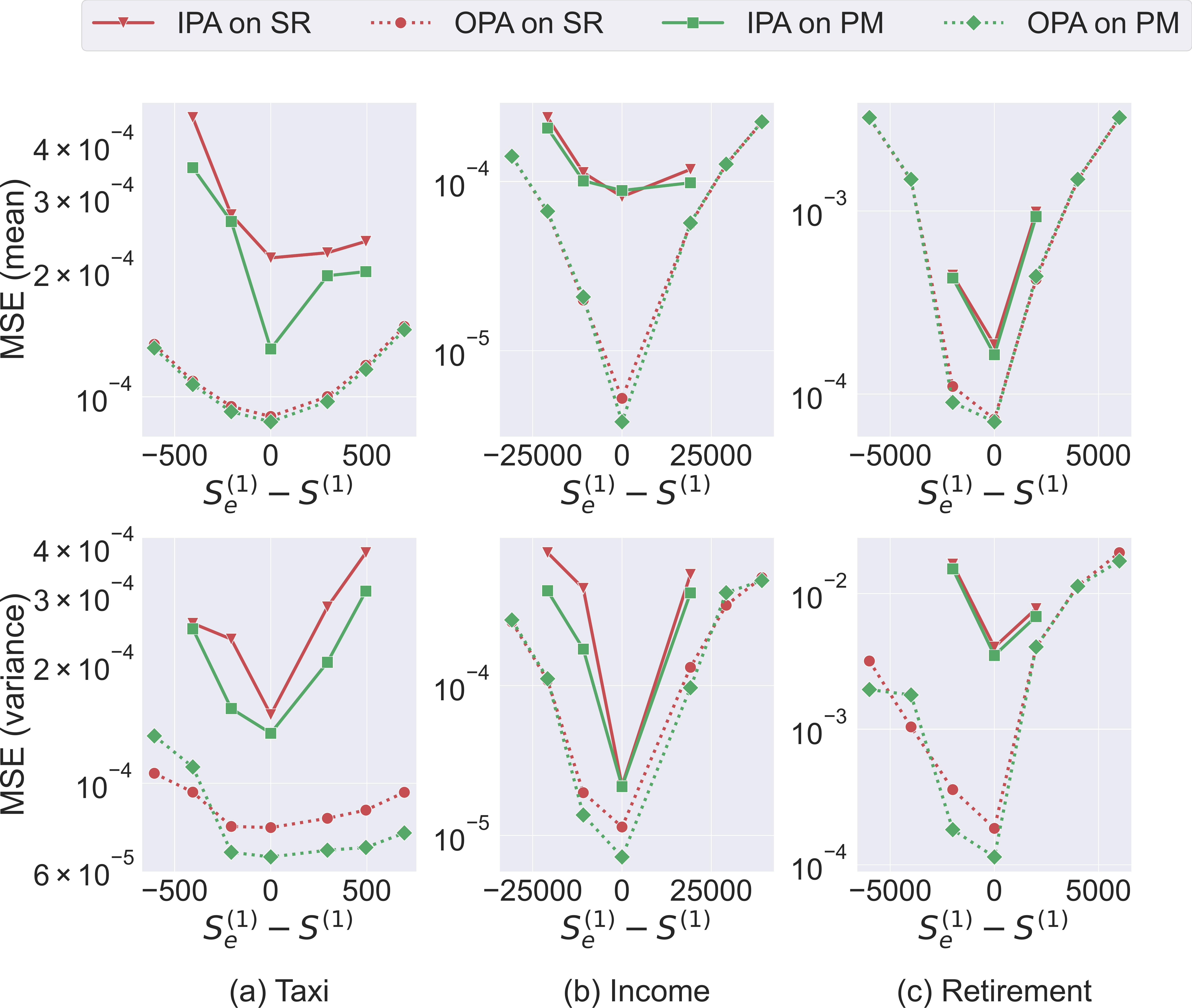}
    \caption{Attack error in SR and PM, varying $S_{e}^{(1)}$. Target values are $\mu_{t_1}$ and $\sigma^2_{t_1}$, $\epsilon = 1$, $m = 0.1n$, $S_{e}^{(2)} = S_{e}^{*(2)}$, $n_e = n_e^*$.}
    \label{MSE_target_mean_var_Xe_composition}
    \vspace{-10pt}
\end{figure}

\begin{itemize}[leftmargin=*]
\item Figures \ref{MSE_target_mean_var_Xe_composition} and \ref{MSE_target_mean_var_XSe_composition} show that \textsf{OPA} outperforms \textsf{IPA} when attacking mean and variance across all datasets.

\item The more accurate the estimation of $S_{e}^{(1)}$ is, the smaller the attack error is as shown in Figure~\ref{MSE_target_mean_var_Xe_composition}. The MSEs of both \textsf{OPA} and baseline are approximately symmetrical about the estimation error $S^{(1)} - S_{e}^{(1)} = 0$, because the term $(S^{(1)} - S_{e}^{(1)})^2$ in the attack error computation increases when $S_{e}^{(1)}$ moves farther away from $S^{(1)}$ (see Table~\ref{tab:attack_error}). 

\item In Figure~\ref{MSE_target_mean_var_XSe_composition}, we observe the similar impact of $S_{e}^{(2)}$ on variance. However, the impact on the mean shows a different trend. With increased $S_{e}^{(2)}$, the MSE of \textsf{IPA} against SR grows but reduces against PM. This is reasonable because when $S_{e}^{(2)}$ grows, the sum of the squared fake values $\sum_{i=1}^{m} y_i^2$ decreases, leading to a large error in SR and a small error in PM (Lemma~\ref{mechanism_error}). The MSE of \textsf{OPA} remains almost constant with varying $S_{e}^{(2)}$ because the attacker crafts fake values in $g_1$ to manipulate the mean, which is not affected by $S^{(2)}$.

\item Due to the constrained sufficient condition for launching \textsf{IPA}, the default $\beta$ is insufficient when both $S^{(1)}_e$ and $S^{(2)}_e$ are far away from $S^{(1)}$ and $S^{(2)}$ respectively, thus no corresponding MSE recorded in Figure~\ref{MSE_target_mean_var_Xe_composition} and \ref{MSE_target_mean_var_XSe_composition}.
\end{itemize}

The recovered value on the server gets closer to the target with a smaller estimation error. In practice, the interpretation of the attack efficacy is subject to the attacker and may vary depending on applications. For example, we in the experiment observe that when an \textsf{OPA} attacker sets the target mean to $0.06$ and target variance to $0.33$ on \textit{Taxi} with 18\% estimation error about $S_{e}^{(1)}$, the recovered mean and variance by SR are $0.056$ (about 3\% from the target) and $0.329$ (about $0.9\%$ from the target) respectively, which may still be considered a success by the attacker. 





\vspace{3pt}
\noindent\textbf{Impact of $n_e$.}  \label{section_impact_ne} We can derive a similar conclusion regarding the relationship between estimation error $n_e-n$ and attack error in Figure~\ref{MSE_target_mean_var_ne_composition}. In general, less estimation error results in better attack performance. The MSEs of \textsf{IPA} and \textsf{OPA} are almost symmetric about $n_e-n$. Besides, \textsf{OPA} performs better in both SR and PM due to more LDP noise introduced in \textsf{IPA}.

    




Similar to estimating $S_{e}^{(1)}$ as analyzed previously, the attacker may not be able to get an accurate estimate of the user number in practice, which will cause the recovered statistics to deviate from the intended values. Again, the interpretation of the deviation here is subject to the attacker's objective. Our experiment reports that given the target mean -0.86, variance 0.02, $10\%$ estimation error of user number ($88,000$ estimated vs. $97,220$ actual) on \textit{Retirement}, a server using SR mechanism can recover the mean and variance to -0.861 (about $0.36\%$ accuracy loss) and 0.0202 (about $4\%$ accuracy loss) respectively, under our output poisoning attack. 

\begin{figure}[!tb]
    \centering
    \includegraphics[scale = 0.16]{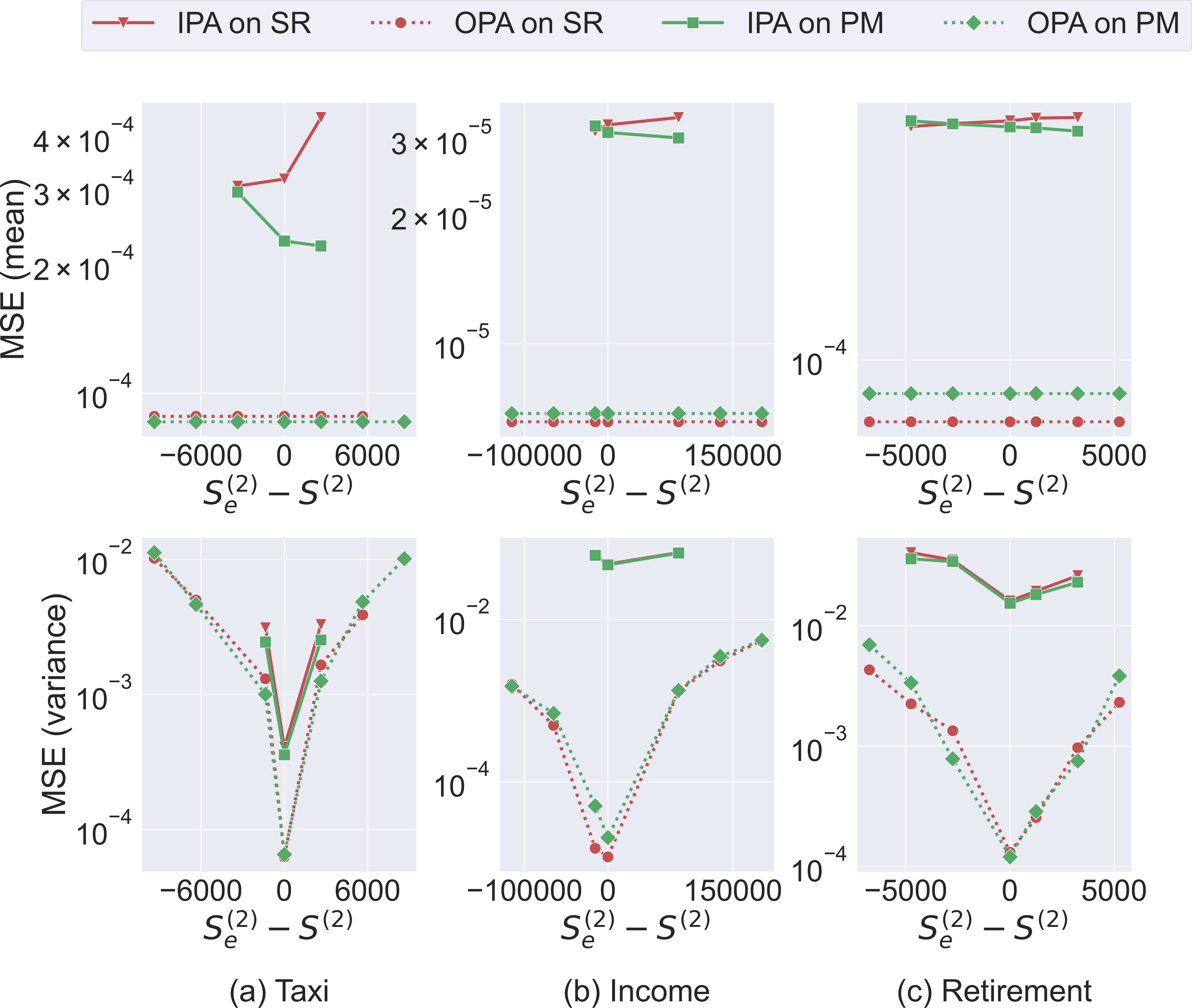}
    \caption{Attack error in SR and PM, varying $S_{e}^{(2)}$. Target values are $\mu_{t_1}$ and $\sigma^2_{t_1}$, $\epsilon = 1$, $m = 0.1n$, $S_{e}^{(1)} = S_{e}^{*(1)}$, $n_e = n_e^*$.}
    \label{MSE_target_mean_var_XSe_composition}
    \vspace{-10pt}
\end{figure}

\begin{figure}[!tb]
    \centering
    \includegraphics[scale = 0.16]{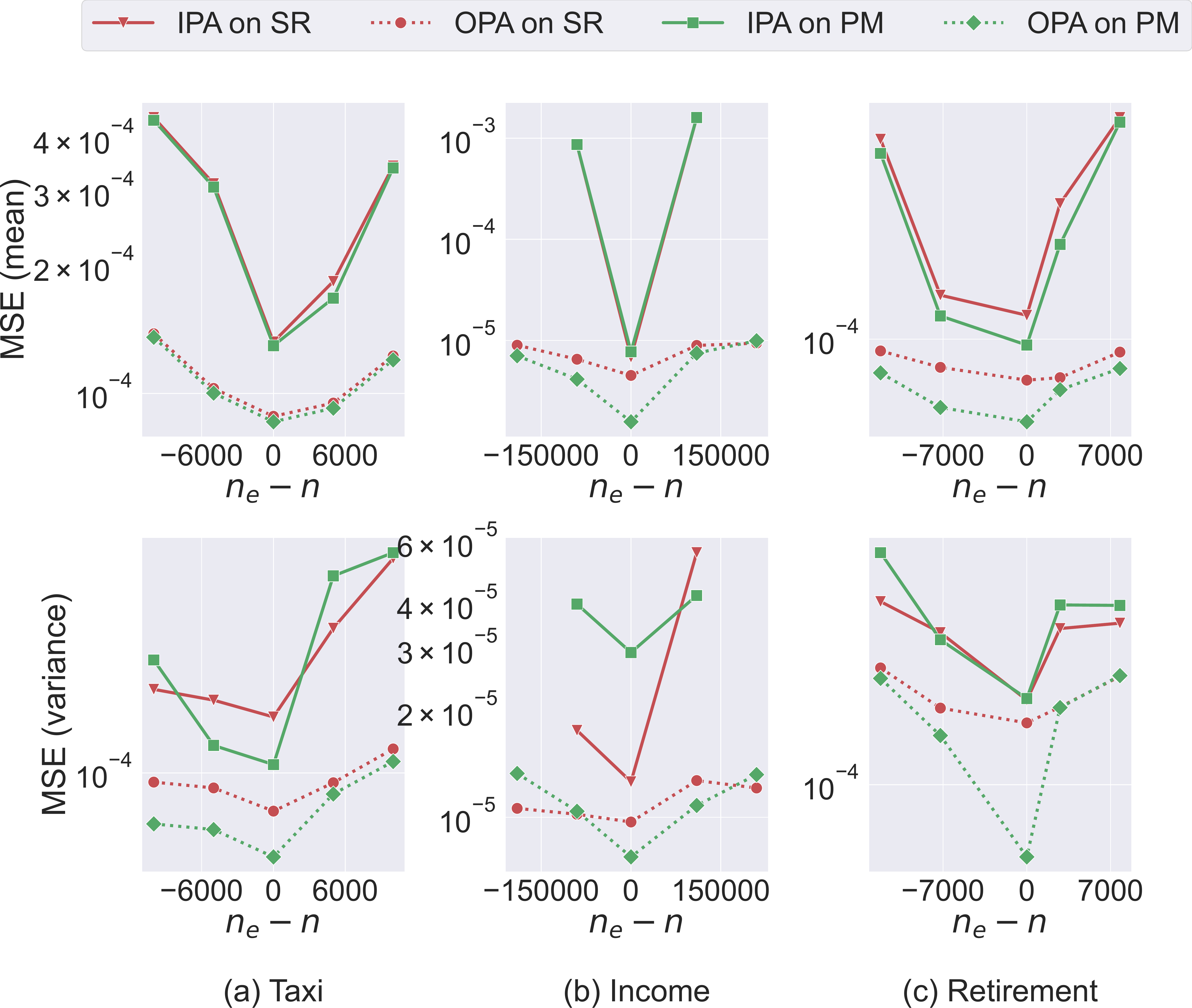}
    \caption{Attack error in SR, varying $n_e$. Target values are $\mu_{t_1}$ and $\sigma^2_{t_1}$, $\epsilon = 1$, $m = 0.1n$, $S_{e}^{(1)} = S_{e}^{*(1)}$, $S_{e}^{(2)} = S_{e}^{*(2)}$.}
    \label{MSE_target_mean_var_ne_composition}
    \vspace{-12pt}
\end{figure}





\begin{figure}[!tb]
    \centering
    \includegraphics[scale = 0.16]{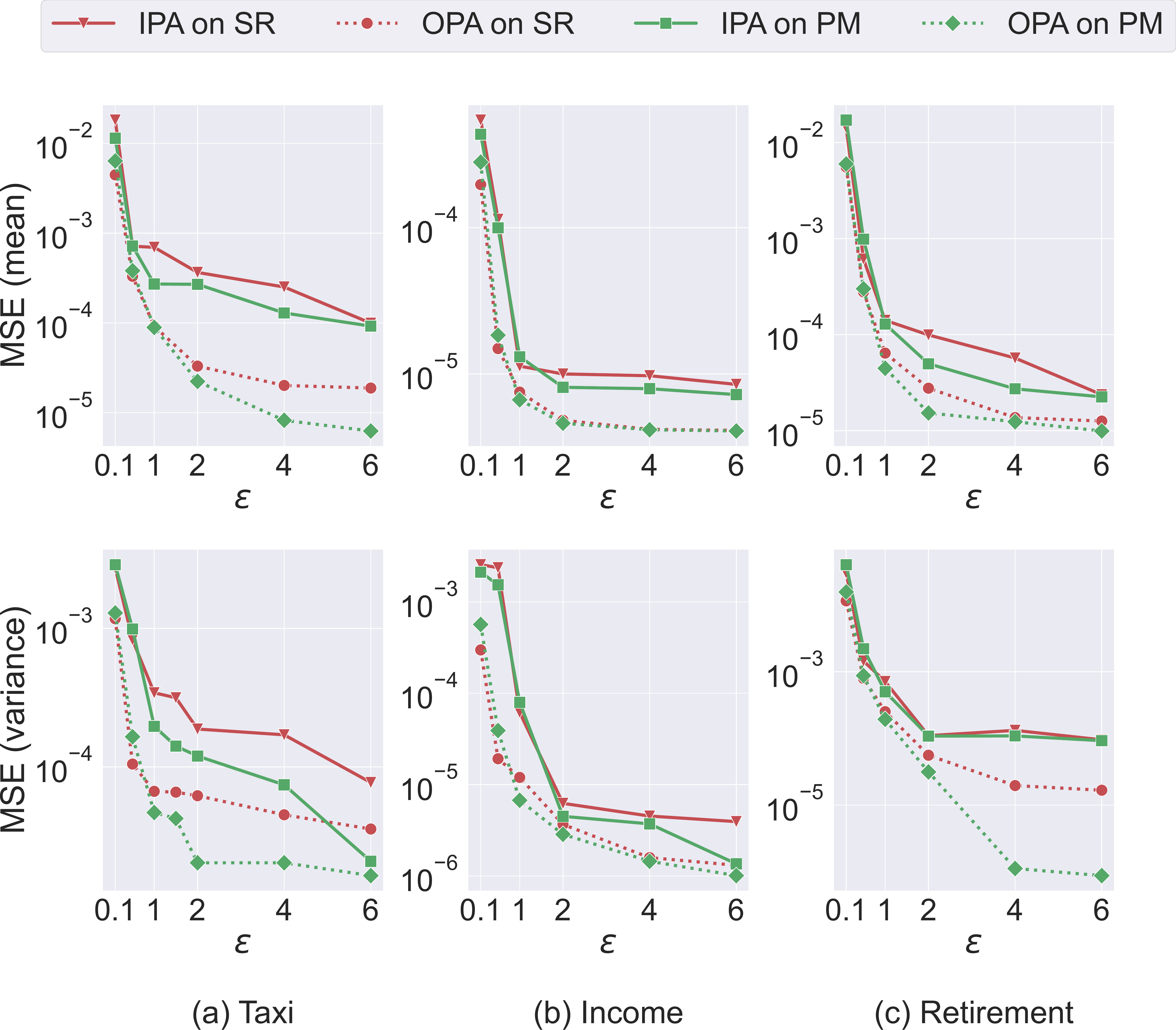}
    \caption{Attack error in SR and PM, varying $\epsilon$. Target values are $\mu_{t_1}$ and $\sigma^2_{t_1}$, $m = 0.1n$, $S_{e}^{(1)} = S_{e}^{*(1)}$, $S_{e}^{(2)} = S_{e}^{*(2)}$, $n_e = n_e^*$.}
    \label{MSE_target_mean_var_eps_combination}
    \vspace{-10pt}
\end{figure}

\vspace{-5pt}
\subsubsection{Other Factors}
We further study the impact of the $\epsilon$, $\beta$ and time cost for the attack. The results confirm the better performance of \textsf{OPA}, the security-privacy consistency and the attack efficiency.

\vspace{2pt}
\noindent\textbf{Impact of $\epsilon$.} 
Figure \ref{MSE_target_mean_var_eps_combination} shows how $\epsilon$ affects attack performance. We empirically confirm the privacy-security consistency with our attacks, which complements theoretical analysis in Section~\ref{security_privacy_consistency}. Overall, the attack performance improves with large $\epsilon$. For the attack on mean and variance, \textsf{OPA} exceeds \textsf{IPA} under all selected $\epsilon$ since \textsf{OPA} is partially influenced by LDP obfuscation. As $\epsilon$ increases, the attack error in PM is smaller than that in SR because PM adds less LDP noise (see Lemma \ref{mechanism_error}). However, even when $\epsilon=0.1$, the \textsf{OPA} results are still close to the targets, i.e. the accuracy loss is less than $2.2\%$ for the target mean and $4\%$ for the target variance on \textit{Retirement}.

\vspace{2pt}
\noindent\textbf{Impact of $\beta$.}  \label{section_impact_beta}
Figure~\ref{MSE_target_mean_var_m_composition} shows the impact of the number of fake users. In reality, an appropriate $\beta$ depends on the application and the resources accessible to the attacker. We varied $\beta$ from 0.05 to 0.8 to comprehensively evaluate the attack by covering extreme cases. 
We observe that the attack errors of both baseline and \textsf{OPA} on mean and variance reduce as $\beta$ grows. This is because the number of fake users is in the denominator of the error calculation for both \textsf{IPA} and \textsf{OPA}. However, \textsf{OPA} performs much better since \textsf{OPA} is only partially affected by LDP noise. For the default $\beta = 0.1$ on \textit{Retirement}, the accuracy loss of \textsf{OPA} toward target mean and variance is only about $0.03\%$ and $0.93\%$ respectively.


\vspace{2pt}
\noindent\textbf{Time Efficiency.} 
If the required estimates have been done in advance, the time cost for the attacks only depends on the calculation of the fake values to be injected. Per our experiment, performing \textsf{OPA} is faster (i.e. less than 0.2 seconds) than \textsf{IPA} (about 10 seconds). This is because the  \textsf{OPA} can directly produce fake values given the explicit expression of $\Psi(y)$ and $\Psi(y^2)$ while the optimization problem in Equation~\eqref{fake_value_optimization} needs to be solved for \textsf{IPA}. We adopted PyTorch in our experiment. Other optimizers may lead to different results for \textsf{IPA}.

\begin{figure}[!tb]
    \centering
    \includegraphics[scale = 0.16]{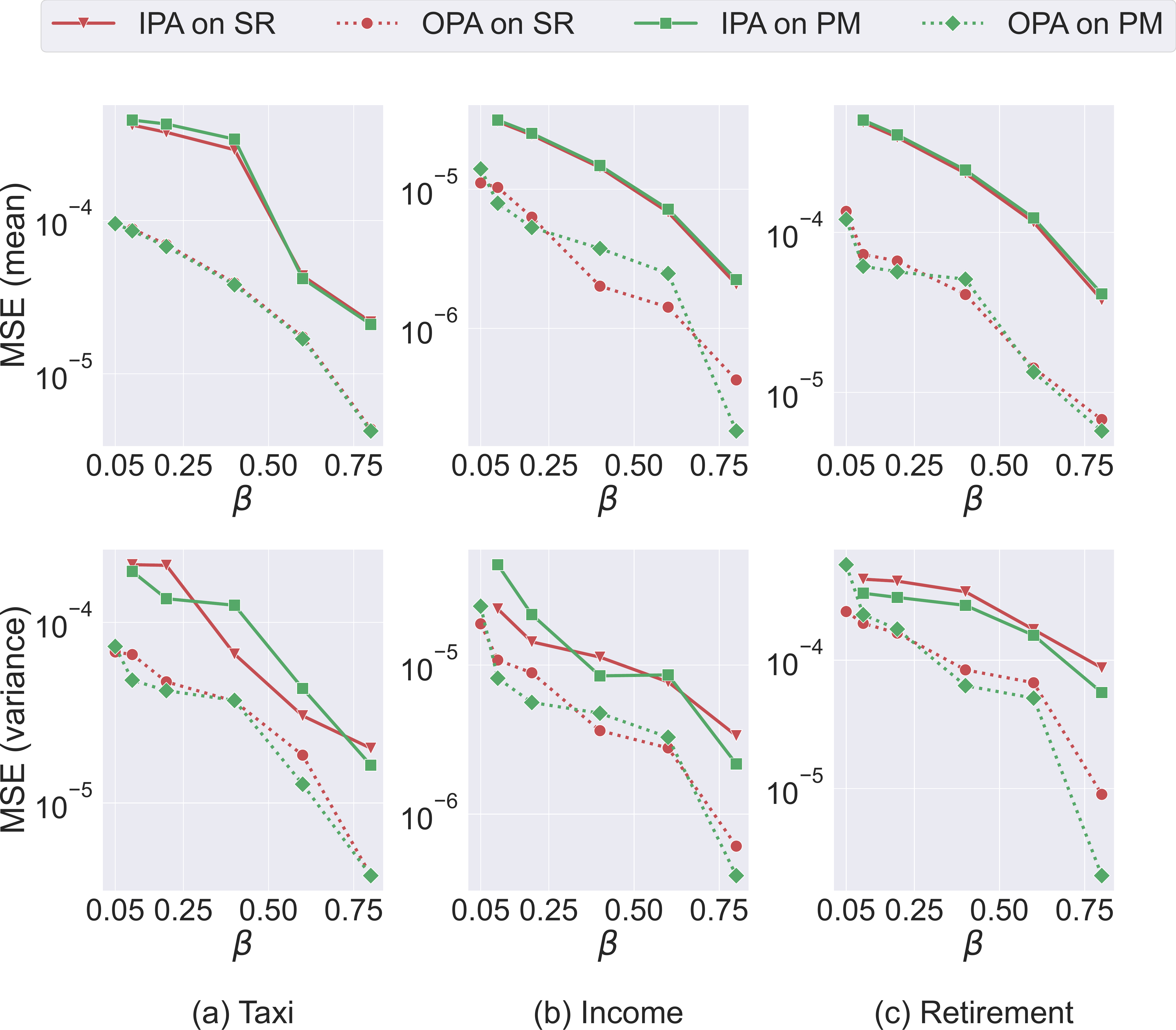}
    \caption{Attack error in SR and PM, varying $\beta$. Target values are $\mu_{t_1}$ and $\sigma^2_{t_1}$, $\epsilon = 1$, $S_{e}^{(1)} = S_{e}^{*(1)}$, $S_{e}^{(2)} = S_{e}^{*(2)}$, $n_e = n_e^*$}
    \label{MSE_target_mean_var_m_composition}
    \vspace{-10pt}
\end{figure}

\section{Mitigation}\label{mitigation}
There are two types of methods proposed in prior research to defend against the data poisoning attack, i.e. normalization \cite{cao2019data} and fake user detection \cite{cao2019data, wu2021poisoning}. The idea of normalizing the LDP estimates to reduce attack effectiveness is based on the na\"ive observation of frequency consistency \cite{wang2020locally}, which is not applicable to mean/variance estimation. Detecting fake users is possible if the fake values are distinguishable from normal values. We present a countermeasure that can tolerate the skewed values and recover the true estimate. 
Contrary to prior work that assumes the server knows user values and the fraction of genuine users as ground truth \cite{wu2021poisoning}, we consider these conditions are difficult to satisfy in reality and our defense does not rely on them.

\vspace{-5pt}
\subsection{Clustering-based Defense} \label{section_mitigation}
We adopt a sampling-then-clustering method to defend against our output poisoning attack, inspired by \cite{cao2021provably} in the context of federated learning. The main idea is to sample multiple subsets of users and then use a clustering algorithm, such as $k$-means, to form two clusters. The cluster that contains more subsets will be used for estimation, while the other will be discarded. The intuition is that since the majority of users are genuine, the mean of most subsets should be similar and close to the true mean. 
More precisely, we first define a sampling rate $r\ (0 < r < 1)$ and derive all $n_1 \choose r n_1$ possible subsets in $g_1$ and all $n_2 \choose r n_2$ possible subsets in $g_2$ without replacement, where $n_i$ is the number of users in $g_i$.
Next, we estimate $\mathbb{E}(x)$ and $\mathbb{E}(x^2)$ for each subset and feed them into $k$-means for $g_1$ and $g_2$. By identifying the benign clusters in $g_1$ and $g_2$, we use their respective cluster centers as $\mathbb{E}(x)$ or $\mathbb{E}(x^2)$ for mean and variance estimation. Our defense could be further optimized by leveraging advanced fault tolerance results \cite{gupta2020fault, liu2021approximate}, which will be left as an important future work.

\begin{figure}[!tb]
    \centering
    \includegraphics[scale=0.19]{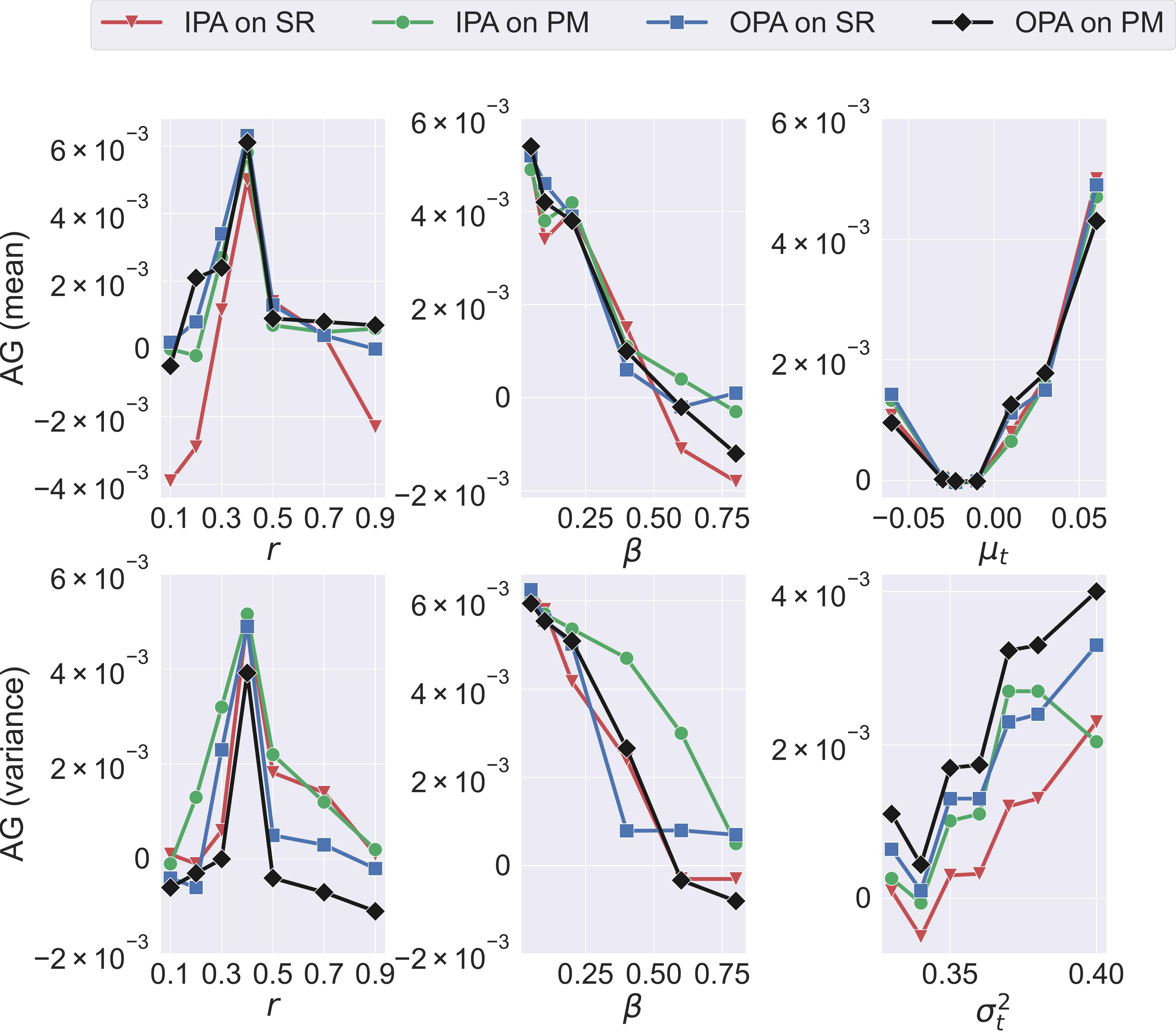}
    \caption{Defense evaluation results. True $\mu = -0.22$ and $\sigma^2 = 0.34$. The default target values are $\mu_{t_1}$ and $\sigma^2_{t_1}$, $\epsilon=1$, $\beta=0.1$, $S_{e}^{(1)} = S_{e}^{*(1)}$, $S_{e}^{(2)} = S_{e}^{*(2)}$, $n_e = n_e^*$.}
    \label{countermeasure_result}
    \vspace{-16pt}
\end{figure}

\vspace{3pt}
\noindent\textbf{Results.} We evaluate the defense performance by measuring the \textit{accuracy gain} (AG) after applying our mitigation.
AG measures the estimation error change before and after the proposed defense as $MSE_{before} - MSE_{after}$. A positive AG value indicates the defense helps the LDP regain the data accuracy after the attack while a negative one represents the ineffectiveness of the mitigation. Thus, the larger AG is, the more effective our defense is against the attack. In practice, the number of subsets $n_i \choose r n_i$ could be too large for an efficient response. We observe that further randomly choosing a small portion (e.g. 5,000 subsets in our experiment) still results in an effective defense.

We use the dataset \textit{Taxi} for result demonstration in Figure~\ref{countermeasure_result}.  It shows that the choice of $r$ will affect the performance (the leftmost figures). A small sampling rate will lead to a small subset, which in turn introduces more bias. On the other hand, a large $r$ results in fewer subsets, but each subset may contain more fake users, thus being subject to manipulation. This may explain why negative AG is observed when $r$ is chosen either too small or large. 
We empirically find an optimal $r$ for the rest of the evaluation. The defense performance is also related to the ratio $\beta$ (the middle figures). When the fraction of fake users is small, our defense is very effective.
When the target value is far from the true value, it is easier to identify the fake values as outliers with increased AG (the rightmost figures). 
Therefore, our method is preferred when the fraction of fake users is small and the attacker wants to skew the mean and variance substantially, which is reasonable for most attacking scenarios. 
 

\vspace{-5pt}
\subsection{Other Defenses}
For an attack that aims to falsify the input of a local LDP instance, such as the baseline \textsf{IPA}, an authenticated data feed system may assist in reestablishing trust for data sources. The current solution, however, is limited to well-known entities \cite{zhang2016town}. Authenticating data from unidentified sources in a distributed environment remains an open problem \cite{kato2021preventing,zhang2016town}. To defend against data poisoning attacks for frequency and heavy hitter estimations, two cryptography-based methods were proposed in \cite{kato2021preventing} and \cite{naor2019not} respectively. Kato \textit{et al.} \cite{kato2021preventing} utilized cryptographic randomized response as a building block to make the LDP protocol for frequency estimation verifiable to the data curator. In \cite{naor2019not}, the multi-party computation was also leveraged to restrict the attacker's capability to manipulate the heavy hitter results. These countermeasures could be used to mitigate the data poisoning attacks in \cite{cao2019data, wu2021poisoning}, but are not applicable to our attacks due to the different LDP perturbation mechanisms for mean and variance estimations. Other potential solutions include adopting hardware-assisted trusted computing \cite{armv9,lee2020keystone} to ensure the authenticity of the LDP protocol execution and communication with the remote server. But this may incur additional costs for software-hardware co-design and security auditing in order to avoid a wide spectrum of side-channel attacks~\cite{van2018foreshadow,schwarz2017malware,gotzfried2017cache,weichbrodt2016asyncshock,murdock2020plundervolt}. 

\section{Related Work}
\noindent\textbf{Data Poisoning Attack to LDP Protocols.} 
Recent research found that LDP is vulnerable to data poisoning attacks. It was shown that the LDP aggregator is sensitive to distribution change of perturbed data. Thus, the result accuracy of non-interactive LDP protocols can be degraded by injecting false data \cite{cheu2019manipulation}. It is further demonstrated that the attacker can maximally deviate the estimate of an item of interest (e.g. categorical data in \cite{cao2019data} and key-value data in \cite{wu2021poisoning}) from the ground truth by formulating the attack as an optimization problem to maximize the attack gain. The solution is the fake data that the attacker will send to the data collector.


We consider a fine-grained attack in this paper, where the attacker aims to control the estimate to some desired values. In general, the above maximal deviation attacks \cite{cao2019data, wu2021poisoning} can be deemed as the extreme cases of our attack. With this new capability, the attacker can target more scenarios for precise result manipulation. In addition, we provide important new insights into the attack impact on LDP and mitigation design.

\vspace{2pt}
\noindent\textbf{Adopting DP/LDP against Poisoning Attacks in ML.} A line of work \cite{naseri2022local, ma2019data, borgnia2021dp} studied using DP/LDP to improve the robustness of machine learning models against data poisoning attacks, where the attacker prepares a poisoned training dataset to change the model behavior to some desired one. In \cite{borgnia2021dp}, differentially private data augmentation was studied for attack mitigation. \cite{ma2019data} investigated the attacks on the ML model trained with differentially private algorithms. Other than central DP, \cite{naseri2022local} further studied the impact of LDP on the defense and observed varying levels of protection-utility trade-offs. Our work has a distinct motivation, i.e., we study the intrinsic robustness and security of LDP protocols in the presence of fine-grained data manipulation. The results of our work may shed light on the related security discussions of using LDP for defenses. For example, a strong LDP perturbation may help reinforce the defense effect in \cite{naseri2022local} while the attack in a central DP setting \cite{naseri2022local, ma2019data, borgnia2021dp} is intuitively analogous to our baseline attack where only input values can be crafted.


\vspace{-10pt}
\section{Discussion} \label{discussion}
There exist other LDP protocols supporting mean/variance estimation \cite{dwork2006calibrating,wang2017locally,li2020estimating}, to which the baseline and our attack are still applicable. \textsf{IPA} is straightforward since it is independent of LDP implementation. For \textsf{OPA}, the attacker can craft fake data in the output domain of the perturbation by leveraging the LDP knowledge. Note that since the aggregation $\Phi()$ is iterative in \cite{li2020estimating}, we cannot derive an explicit mathematical expression to determine fake values in the same way as in this work (e.g., Equations \eqref{output_value_mean} and \eqref{output_value_var}). However, the attacker may obtain a valid solution by simulating the iteration and searching the output domain of the perturbation.

\vspace{2pt}
\noindent\textbf{Frequency Estimation under Pure LDP} \cite{wang2017locally}: The \textsf{IPA} and \textsf{OPA} could be adapted to attack the pure LDP protocols for frequency, such as kRR \cite{duchi2013local}, OUE and OLH \cite{wang2017locally}. The attacker needs to estimate the frequencies of items and inject bogus data as per the the intended frequency of target items. \textsf{OPA} may leverage the LDP protocols to improve performance.

\vspace{2pt}
\noindent\textbf{Distribution Estimation} \cite{li2020estimating}: Distribution estimation can be considered a frequency oracle in the numerical domain, to which our attacks may still be applicable. We provide the attack intuition here. In general, the attacker begins by estimating the original data distribution. Given this, the attack may generate fake data points equal to a specific value $\bm{x}$ to increase the probability density of $\bm{x}$ to the target value. To reduce the probability density, the attacker could provide data that is not equal to $\bm{x}$. 
    
\vspace{2pt}
\noindent\textbf{Graph data mining} \cite{ye2020towards, imola2021locally}: In graph data mining, LDP protocols focus on calculating graph statistics, e.g., counting triangles and $k$-stars in the graph, the degree and adjacency bit vector of each node. We assume in this scenario that the attacker wishes to control the final estimate to some target value. To launch the attack, the attacker could first use a graph generation model, such as BTER \cite{qin2017generating}, to estimate the graph topology. The attacker then could inject bogus nodes and edges into the graph to exert finer control over its statistics.
For all the discussed query types, the security-privacy consistency may remain, as increased privacy introduces additional noise and reduces the effectiveness of the manipulation.

\vspace{-10pt}
\section{Conclusion}
We conducted a systematic study on data poisoning attacks against the LDP protocols for mean and variance estimation. 
We present an effective attack to craft the output of the LDP instance and manipulate both mean and variance estimates according to target values.
The analysis reveals a disturbing fact: the LDP is inherently vulnerable to data poisoning attacks regardless of the privacy budget, i.e., previous attacks are effective when $\epsilon$ is small (high privacy), whereas our attacks perform better when $\epsilon$ is large (low privacy). We also discussed the applicability of our attacks against other query types and shed light on the promising mitigation development.

\section{Acknowledgment}
We would like to thank the Shepherd and anonymous reviewers for their insightful comments and guidance. Hui Li and Xiaoguang Li were supported in part by NSFC 61932015.

\bibliographystyle{plain}
{\footnotesize
\bibliography{ref}
}

\begin{appendix}
\section{Proof of Theorem \ref{bound_y2}} \label{appendix_solution_max_y2}
In the optimization problem, the objective function $-\sum_{i=1}^{m} y_i^2$ is convex, the inequality constraints are continuously differentiable convex functions and the equality constraint is an affine function. Thus we can prove that the solution is the maximum by proving the solution satisfies KKT conditions \cite{karush2014minima, kuhn2014nonlinear}. We prove a general case where the value is in $[a, b]$, and Theorem \ref{bound_y2} can be derived by setting $a=-1$ and $b=1$. Let the function $L(y_1, ..., y_m)$ equal
\begin{align*}
    -\sum_{i=1}^m y_i^2 + \alpha(\sum_{i=1}^m y_i - A)
    + \sum_{i=1}^m\left[\beta_{i}^{(a)}(a-y_i) + \beta_{i}^{(b)}(y_i - b) \right],
\end{align*}
where $A = (n+m)\mu_t - S^{(1)}_e$, and $\alpha$, $\beta_{i}^{(a)}$ and $\beta_{i}^{(b)}$ are constants. The solution should satisfy four sets of conditions:

\noindent\textbf{Stationarity:}
$
    \forall i: \frac{\partial L}{\partial y_i} = -2y_i + \alpha - \beta_{i}^{(a)} + \beta_{i}^{(b)} = 0.
$

\noindent\textbf{Primal feasibility:} $\sum_{i=1}^m y_i - A = 0, \forall i: a-y_i \leq 0, y_i - b \leq 0$.

\noindent\textbf{Dual feasibility:} For any $i$, $\beta_{i}^{(a)}, \beta_{i}^{(b)} \geq 0$.

\noindent\textbf{Complementary slackness:} $\forall i: \beta_{i}^{(a)}(a-y_i) = 0, \beta_{i}^{(b)}(y_i - b) = 0$.

Since the partial derivative of $L$ should be zero, we have
\begin{align*}
    y_i = \frac{1}{2}(\alpha - \beta_{i}^{(a)} + \beta_{i}^{(b)})
    \Rightarrow \frac{1}{2} \sum_i^m (\alpha - \beta_{i}^{(a)} + \beta_{i}^{(b)}) = A.
\end{align*}
Thus we can rewrite $\alpha$ and $y_i$ as $\alpha = \frac{2A + \sum_i^m \beta_{i}^{(a)} - \sum_i^m \beta_{i}^{(b)}}{m}$ and $y_i = \frac{1}{2} (\frac{2A + \sum_i^m \beta_{i}^{(a)} - \sum_i^m \beta_{i}^{(b)}}{m} - \beta_{i}^{(a)} + \beta_{i}^{(b)})$.
Let the domain of $y_i$ be $D = D_a \bigcup D_b \bigcup D_{ab}$ s.t., $\forall y_i \in D_a, y_i = a$, $\forall y_i \in D_b, y_i = b$, and $\forall y_i \in D_{ab}, a < y_i < b$. Given the solution, we have $\left| D_b \right| = \floor{\frac{A - ma}{b - a}}$, $\left| D_a \right| = m - 1 - \left| D_b \right|$ and $\left| D_{ab} \right| = 1$.
For $\forall y_i \in D_a$, we have $\forall i: \beta_i^{(b)} = 0$ due to the complementary slackness, and
\begin{align*}
    &\quad y_i = a = \frac{1}{2} (\frac{2A + \sum_i^m \beta_{i}^{(a)}}{m} - \beta_{i}^{(a)}) \\
    &\Rightarrow \beta_{i}^{(a)} = \frac{2A + (m - 1 - \left|D_b\right|) \beta_i^{(a)}}{m} - 2a
    = \frac{2(A - ma)}{\left|D_b\right| - 1}.
\end{align*}
Since $A \geq ma$ and $\left|D_b\right| \geq 1$, we have $\beta_{i}^{(a)} \geq 0$ for all $y_i \in D_a$. Therefore, for $\forall i: y_i \in D_a$, we have $\beta_{i}^{(a)}, \beta_{i}^{(b)} \geq 0$ and $\beta_{i}^{(a)}(a-y_i) = 0$ due to $y_i = a$, and $\beta_{i}^{(b)}(y_i - b) = 0$.

For $\forall y_i \in D_b$, we have $\forall i: \beta_i^{(a)} = 0$ due to the complementary slackness, and
\begin{align*}
    &\quad y_i = b = \frac{1}{2} (\frac{2A - \sum_i^m \beta_{i}^{(b)}}{m} + \beta_{i}^{(b)}) \\
    &\Rightarrow \beta_{i}^{(b)} = 2b - \frac{2A - \left|D_b\right| \beta_i^{(b)}}{m}
    = \frac{2(mb - A)}{m - \left|D_b\right|}.
\end{align*}
Since $mb \geq A$ and $\left|D_b\right| \leq m$, we have $\beta_{i}^{(b)} \geq 0$ for all $y_i \in D_b$.
Therefore, for $\forall i: y_i \in D_b$, we have $\beta_{i}^{(a)}, \beta_{i}^{(b)} \geq 0$ and $\beta_{i}^{(a)}(a-y_i) = 0$ due to $\beta_{i}^{(a)} = 0$, and $\beta_{i}^{(b)}(y_i - b) = 0$ due to $y_i = b$.
For $y_i \in D_{ab}$, we have $\beta_i^{(a)} = \beta_i^{(b)} = 0$ due to the complementary slackness.

In conclusion, for $\forall i: y_i \in D$, the partial derivative $\frac{\partial L}{\partial y_i}$ is zero (satisfying \textbf{Stationarity}), the sum $\sum_{i=1}^m y_i = A$, and $\forall i: a \leq y_i \leq b$ (satisfying \textbf{Primal feasibility}), the constants $\beta_{i}^{(a)}, \beta_{i}^{(b)} \geq 0$ for all $y_i \in D$ (satisfying \textbf{Dual feasibility}), and $\beta_{i}^{(a)}(a-y_i) = 0, \beta_{i}^{(b)}(y_i - b) = 0$ for all $y_i \in D$ (satisfying \textbf{Complementary slackness}).

\section{Error of IPA on SR}
\label{appendix_proof_error_input_attack_SR}
We first analyze the error of $\hat{\mu}_t$ under the SR mechanism. In SR, the estimated mean $\hat{\mu}_t$ after the attack is
\begin{align*}
    \frac{2}{n + m} \left( \sum_{i=1}^{n_1} \Phi_{1}(\Psi(x_{i, (1)})) + \sum_{i=1}^{m_1} \Phi_{1}(\Psi(y_{i, (1)})) \right) = \hat{\mu}_t.
\end{align*}
Thus we have the expectation of $\hat{\mu}_t$
\begin{align*}
    &\quad \mathbb{E}(\hat{\mu}_t) = \frac{2}{m+n} \mathbb{E}\left[ \sum_{i=1}^{n_1} \Phi_{1}(\Psi(x_{i, (1)})) + \sum_{i=1}^{m_1} \Phi_{1}(\Psi(y_{i, (1)})) \right] \\
    &= \frac{2}{m+n} \mathbb{E}\left[ \mathbb{E}\left[ \sum_{i=1}^{n_1} \Phi_{1}(\Psi(x_{i, (1)})) + \sum_{i=1}^{m_1} \Phi_{1}(\Psi(y_{i, (1)})) \;\middle\vert\; g_1 \right] \right] \\
    &= \frac{2\mathbb{E}\left[ \sum_{i=1}^{n_1} x_{i,(1)} + \sum_{i=1}^{m_1} y_{i,(1)} \right]}{m+n} 
    = \frac{m+n_e}{m+n}\mu_t + \frac{(S^{(1)} - S_{e}^{(1)})}{m+n}.
\end{align*}
Then we can calculate the error as $\mathbb{E}[(\hat{\mu}_t - \mu_t)^2] = Var[\hat{\mu}_t] + (\mathbb{E}(\hat{\mu}_t) - \mu_t)^2$.
The bias is known due to $\mathbb{E}(\hat{u}_t) = \frac{m+n_e}{m+n}\mu_t + \frac{1}{m+n}(S^{(1)} - X_e)$. Here we study the term $Var[\hat{\mu}_t]$. We denote the $\Phi(\Psi())$ by $M()$. Let $\sum_{i=1}^{n_1} \Phi_{1}(\Psi(x_{i, (1)}))$ and $\sum_{i=1}^{m_1} \Phi_{1}(\Psi(y_{i, (1)}))$ be $M(X_{g_1})$ and $M(Y_{g_1})$ respectively. We also denote $\sum_{i=1}^{n_1}x_i$ and $\sum_{i=1}^{m_1}y_i$ as $X_{g_1}$ and $Y_{g_1}$. Thus, we have
\begin{align*}
    &\quad Var[\hat{\mu}_t] = \mathbb{E}\left[ (\hat{\mu}_t - \mathbb{E}(\hat{\mu}_t))^2 \right] \\
    &= \mathbb{E}\Bigg[\Bigg(\frac{2}{n + m} ( M(X_{g_1}) + M(Y_{g_1})) - \frac{2}{m+n}(X_{g_1} + Y_{g_1} \Bigg)^2 \Bigg] \\
    &+ \mathbb{E}\Bigg[\Bigg( \frac{2}{m+n}(X_{g_1} + Y_{g_1}) - \mathbb{E}(\hat{\mu}_t) \Bigg)^2\Bigg] \\
    &+ 2\mathbb{E}\Bigg[ \Bigg(\frac{2}{n + m} \left( M(X_{g_1}) + M(Y_{g_1}) \right) - \frac{2}{m+n}(X_{g_1} + Y_{g_1})\Bigg) \\
    &\times \Bigg( \frac{2}{m+n}(X_{g_1} + Y_{g_1}) - \mathbb{E}(\hat{\mu}_t) \Bigg) \Bigg].
\end{align*}
The variance contains three terms. For the first term, 
\begin{align*}
    &\quad \mathbb{E}\Bigg[\Bigg(\frac{2}{n + m} ( M(X_{g_1}) + M(Y_{g_1})) - \frac{2}{m+n}(X_{g_1} + Y_{g_1}) \Bigg)^2 \Bigg] \\
    &= \frac{2}{(m+n)^2(p-q)^2} \Bigg( m+n - (p-q)^2 \Bigg( S^{(2)} + \sum_{i=1}^{m}y_i^2 \Bigg) \Bigg).
\end{align*}
The first equality is based on Lemma \ref{mechanism_error}. Since $\sum_{i=1}^{m} y_i^2 = (n_e+m)(\mu_t^2 + \sigma_t^2) - S_{e}^{(2)}$, the first term equals
\begin{align*}
    \frac{2}{(m+n)(p-q)^2} - 2\frac{(n_e+m)(\mu_t^2 + \sigma_t^2)}{(m+n)^2} - 2\frac{S^{(2)} - S_{e}^{(2)}}{(m+n)^2}.
\end{align*}
From the standard analysis on sampling process, the second term equals $\frac{n_e+m}{(m+n)^2}(\mu_t^2 + \sigma_t^2) + \frac{1}{(m+n)^2} (S^{(2)} - S_{e}^{(2)})$.
Since $\mathbb{E}[M(X_{g_1})] = \mathbb{E}[X_{g_1}]$, $\mathbb{E}[M(Y_{g_1})] = \mathbb{E}[Y_{g_1}]$ and $\mathbb{E}(\hat{\mu}_t)$ is a constant, we have the third term being zero. Therefore, based on the above three terms, we have the error
\begin{align*}
    &\quad \mathbb{E}[(\hat{\mu}_t - \mu_t)^2] = \left( \frac{n_e - n}{m+n} \mu_t + \frac{(S^{(1)} - S_{e}^{(1)})}{(m+n)} \right)^2 \\
    &+ \frac{2}{(m+n)(p-q)^2} - \frac{(n_e+m)(\mu_t^2 + \sigma_t^2)}{(m+n)^2} - \frac{S^{(2)} - S_{e}^{(2)}}{(m+n)^2}
\end{align*}

Then we study the error of $\hat{\sigma}_t^2$ under the SR mechanism. Denote $\sum_{i=1}^{n_2} \Phi_{2}(\Psi(x_{i, (2)}^2))$ and $\sum_{i=1}^{m_2} \Phi_{2}(\Psi(y_{i, (2)}^2))$ by $M(X_{g_2})$ and $M(Y_{g_2})$, and $\sum_{i=1}^{n_2}x_i^2$ and $\sum_{i=1}^{m_2}y_i^2$ by $X_{g_2}$ and $Y_{g_2}$ respectively. The estimated variance (after attack) can be written as $\frac{2}{n + m} \left( M(X_{g_2}) + M(Y_{g_2}) \right) - \hat{\mu}_t^2 = \hat{\sigma}_t^2$.
Thus we have the expectation of $\hat{\sigma}_t^2$
\begin{align*}
    &\quad \mathbb{E}(\hat{\sigma}_t^2) = \frac{2}{m+n} \mathbb{E}\left[ M(X_{g_2}) + M(Y_{g_2}) \right] - \mathbb{E}[\hat{\mu}_t^2] \\
    &= \frac{m+n_e}{m+n} (\mu_t^2 + \sigma_t^2) + \frac{(S^{(2)} - S_{e}^{(2)})}{m+n} - (Var(\hat{\mu}_t) + \mathbb{E}[\hat{\mu}_t]^2).
\end{align*}
We can calculate the error $ \mathbb{E}[(\hat{\sigma}_t^2 - \sigma_t^2)^2] = Var[\hat{\sigma}_t^2] + (\mathbb{E}(\hat{\sigma}_t^2) - \sigma_t^2)^2$. 
The bias is also known since the expectation $\mathbb{E}[\hat{\sigma}_t^2]$ is known. Next we study the term $Var[\hat{\sigma}_t^2]$
\begin{align*}
    Var[\hat{\sigma}_t^2] = Var\left[ \frac{2}{m+n} ( M(X_{g_2}) + M(Y_{g_2}) ) \right] + Var[\hat{\mu}_t^2]
\end{align*}
Similar to the analysis of $Var[\hat{\mu}_t]$ which is $Var[ \frac{2}{m+n} ( M(X_{g_1}) + M(Y_{g_1}) ) ]$, we denote the $\sum_{i=1}^{n}x_i^4$ by $S^{(4)}$ and have $Var\left[ \frac{2( M(X_{g_2}) + M(Y_{g_2}) )}{m+n}  \right]$ equal $\frac{2}{(m+n)(p-q)^2} - \frac{S^{(4)} + \sum_{i=1}^{m}y_i^4}{(m+n)^2}$.
Since each $y_i^4 \geq 0$, we have 
$
    \sum_{i=1}^{m}y_i^4 \geq 0.
$
For the term $Var[\hat{\mu}_t^2] = \mathbb{E}[\hat{\mu}_t^4] - \mathbb{E}[\hat{\mu}_t^2]^2$, we have $\mathbb{E}[\hat{\mu}_t^2]^2 \geq 0$ and $\hat{\mu}_t \leq b$. Thus, it is bounded by $\quad Var[\hat{\mu}_t^2] = \mathbb{E}[\hat{\mu}_t^4] - \mathbb{E}[\hat{\mu}_t^2]^2 \leq \mathbb{E}[\hat{\mu}_t^4] \leq 1$.

Given $Var[\hat{\mu}_t^2]$, $Var\left[ \frac{2}{m+n} ( M(X_{g_2}) + M(Y_{g_2}) ) \right]$ and $\mathbb{E}(\hat{\sigma}_t^2)$, we have the upper bound of error $\mathbb{E}[(\hat{\sigma}_t^2 - \sigma_t^2)^2]$.

\section{Error IPA on PM}
\label{appendix_proof_error_input_attack_PM}
Since the proof is the same as \textsf{IPA} against SR, we omit the details and use the same set of notations. We first analyze the error of $\hat{\mu}_t$ in PM. 
The expectation of $\hat{\mu}_t$ is
\begin{align*}
    \mathbb{E}(\hat{\mu}_t) 
    = \frac{m+n_e}{m+n}\mu_t + \frac{1}{m+n}(S^{(1)} - S_{e}^{(1)}).
\end{align*}

Then we can calculate the error $\mathbb{E}[(\hat{\mu}_t - \mu_t)^2] = Var[\hat{\mu}_t] + (\mathbb{E}(\hat{\mu}_t) - \mu_t)^2$.
The bias is known due to $\mathbb{E}(\hat{u}_t) = \frac{m+n_e}{m+n}\mu_t + \frac{1}{m+n}(S^{(1)} - S_{e}^{(1)})$. 
We expand the variance $Var[\hat{\mu}_t]$ to the same three terms as in the analysis of \textsf{IPA} against SR. Based on the Lemma \ref{mechanism_error} and $\sum_{i=1}^{m} y_i^2 = (n_e+m)(\mu_t^2 + \sigma_t^2) - S_{e}^{(2)}$, we have the first term equal
\begin{align*}
    \frac{2 (e^{\epsilon/2} + 3)}{3(n+m)(e^{\epsilon/2} - 1)^2} + \frac{2 ((n_e+m)(\mu_t^2 + \sigma_t^2) + (S^{(2)} - S_{e}^{(2)}))}{(n+m)^2 (e^{\epsilon/2} - 1)}.
\end{align*}
From the standard analysis on sampling process, the second term equals
\begin{align*}
    \frac{m+n_e}{(m+n)^2}(\mu_t^2 + \sigma_t^2) + \frac{1}{(m+n)^2} (S^{(2)} - S_{e}^{(2)}).
\end{align*}

Since $\mathbb{E}[M(X_{g_1})] = \mathbb{E}[X_{g_1}]$, $\mathbb{E}[M(Y_{g_1})] = \mathbb{E}[Y_{g_1}]$ and $\mathbb{E}(\hat{\mu}_t)$ is a constant, we have the third term being zero. Therefore, given the above three terms, we have the error
\begin{align*}
    &\quad \mathbb{E}[(\hat{\mu}_t - \mu_t)^2] \\
    &= \frac{2(e^{\epsilon/2} + 3)}{3(n+m)(e^{\epsilon/2} - 1)^2} + \left( \frac{n_e - n}{m+n} \mu_t + \frac{(S^{(1)} - S_{e}^{(1)})}{(m+n)} \right)^2 \\
    &+ \frac{m+n_e}{(m+n)^2}(\mu_t^2 + \sigma_t^2) + \frac{1}{(m+n)^2} (S^{(2)} - S_{e}^{(2)}) \\
    &+ \frac{2 ((n_e+m)(\mu_t^2 + \sigma_t^2) + (S^{(2)} - S_{e}^{(2)}))}{(n+m)^2 (e^{\epsilon/2} - 1)}.
\end{align*}

We now study the error of $\hat{\sigma}_t^2$ under PM mechanism. 
Similar to the analysis of \textsf{IPA} on SR, the expectation of $\hat{\sigma}_t^2$ is
\begin{align*}
    \frac{m+n_e}{m+n}(\mu_t^2 + \sigma_t^2) + \frac{(S^{(2)} - S_{e}^{(2)})}{m+n} - (Var(\hat{\mu}_t) + \mathbb{E}[\hat{\mu}_t]^2).
\end{align*}
We calculate the error as
\begin{align*}
    \mathbb{E}[(\hat{\sigma}_t^2 - \sigma_t^2)^2] = Var[\hat{\sigma}_t^2] + (\mathbb{E}(\hat{\sigma}_t^2) - \sigma_t^2)^2.
\end{align*}
The bias is known since the expectation $\mathbb{E}[\hat{\sigma}_t^2]$ is known. Next we study the term $Var[\hat{\sigma}_t^2]$,
\begin{align*}
    Var[\hat{\sigma}_t^2] = Var\left[ \frac{2}{m+n} ( M(X_{g_2}) + M(Y_{g_2}) ) \right] + Var[\hat{\mu}_t^2].
\end{align*}
Similar to the analysis of $Var[\hat{\mu}_t]$ which is $Var[ \frac{2}{m+n} ( M(X_{g_1}) + M(Y_{g_1}) ) ]$, we denote the $\sum_{i=1}^{n}x_i^4$ by $S^{(4)}$ and have
\begin{align*}
    &\quad Var\left[ \frac{2( M(X_{g_2}) + M(Y_{g_2}) )}{m+n} \right] = \frac{2 (e^{\epsilon/2} + 3)}{3(n+m)(e^{\epsilon/2} - 1)^2} \\
    &+ \frac{2 (S^{(4)} + \sum_{i=1}^{m}y_i^4)}{(n+m)^2 (e^{\epsilon/2} - 1)} 
    + \frac{\left( S^{(4)} + \sum_{i=1}^{m}y_i^4 \right)}{(m+n)^2}.
\end{align*}
Since each $y_i^2 \geq 0$, we have $\sum_{i=1}^{m}y_i^4$ is less than $(\sum_{i=1}^{m}y_i^2)^2$, which equals 
\begin{align*}
    (n_e+m)^2(\mu_t^2 + \sigma_t^2)^2 + S_{e}^{(2)2} - 2(n_e+m)(\mu_t^2 + \sigma_t^2)S_{e}^{(2)}.
\end{align*}

For the term $Var[\hat{\mu}_t^2] = \mathbb{E}[\hat{\mu}_t^4] - \mathbb{E}[\hat{\mu}_t^2]^2$, we have $\mathbb{E}[\hat{\mu}_t^2]^2 \geq 0$ and $\hat{\mu}_t \leq b$. Thus, it is bounded by
\begin{align*}
    \quad Var[\hat{\mu}_t^2] &= \mathbb{E}[\hat{\mu}_t^4] - \mathbb{E}[\hat{\mu}_t^2]^2 \leq \mathbb{E}[\hat{\mu}_t^4] \leq 1.
\end{align*}

Given $Var[\hat{\mu}_t^2]$, $Var\left[ \frac{2}{m+n} ( M(X_{g_2}) + M(Y_{g_2}) ) \right]$ and $\mathbb{E}(\hat{\sigma}_t^2)$, we have the upper bound of error $\mathbb{E}[(\hat{\sigma}_t^2 - \sigma_t^2)^2]$.

\section{Error of OPA on SR}
\label{appendix_proof_error_output_attack_SR}
We first analyze the error of $\hat{\mu}_t$ under the SR mechanism. The estimated mean $\hat{\mu}_t$ after the attack is
\begin{align*}
    \frac{2}{n + m} \left( \sum_{i=1}^{n_1} \Phi_{1}(\Psi(x_{i, (1)})) + \sum_{i=1}^{m_1} \Phi_{1}(\Psi(y_{i, (1)})) \right) = \hat{\mu}_t.
\end{align*}
Thus the expectation of $\hat{\mu}_t$ is $\mathbb{E}(\hat{\mu}_t) = \frac{m+n_e}{m+n}\mu_t + \frac{1}{m+n}(S^{(1)} - S_{e}^{(1)})$.
The error is calculated as $\mathbb{E}[(\hat{\mu}_t - \mu_t)^2] = Var[\hat{\mu}_t] + (\mathbb{E}(\hat{\mu}_t) - \mu_t)^2.$
The bias is known due to $\mathbb{E}(\hat{u}_t) = \frac{m+n_e}{m+n}\mu_t + \frac{1}{m+n}(S^{(1)} - S_{e}^{(1)})$. Here we study the term $Var[\hat{\mu}_t]$. We denote the $\Phi(\Psi())$ by $M()$, and $\sum_{i=1}^{n_1} \Phi_{1}(\Psi(x_{i, (1)}))$ and $\sum_{i=1}^{m_1} \Phi_{1}(\Psi(y_{i, (1)}))$ by $M(X_{g_1})$ and $M(Y_{g_1})$. Let $\sum_{i=1}^{n_1}x_i$ and $\sum_{i=1}^{m_1}y_i$ be $X_{g_1}$ and $Y_{g_1}$. We have
\begin{align*}
    Var[\hat{\mu}_t] = Var\left[ \frac{2}{n + m} \left( M(X_{g_1}) + M(Y_{g_1}) \right) \right].
\end{align*}
Since the adversary directly crafts the output values, the term $M(Y_{g_1})$ is a constant, which can be ignored in the variance. Therefore, the variance
\begin{align*}
    Var[\hat{\mu}_t] 
    = \frac{4}{(m+n)^2} \mathbb{E}\left[ \left( M(X_{g_1}) - X_{g_1} + X_{g_1} - \frac{S^{(1)}}{2} \right)^2 \right].
\end{align*}
We then calculate the expected value,
\begin{align*}
    &\quad \mathbb{E}\left[ \left( M(X_{g_1}) - X_{g_1} + X_{g_1} - \frac{S^{(1)}}{2} \right)^2 \right] \\
    &= \mathbb{E}\left[\left( M(X_{g_1}) - X_{g_1} \right)^2 \right] + \mathbb{E}\left[\left( X_{g_1} - \frac{S^{(1)}}{2} \right)^2\right] \\
    &+ 2\mathbb{E}\left[ \left( M(X_{g_1}) - X_{g_1} \right) \times \left( X_{g_1}  - \frac{S^{(1)}}{2}) \right) \right].
\end{align*}
It contains three terms. For the first term, 
\begin{align*}
    &\quad \mathbb{E}\left[\left(M(X_{g_1}) - X_{g_1} \right)^2 \right] 
    = \mathbb{E}\left[ \mathbb{E}\left[ \left(M(X_{g_1}) - X_{g_1} \right)^2 \mid g_1 \right] \right] \\
    &= \mathbb{E}\Bigg[ \frac{1}{(p-q)^2} \Bigg( n_1 - (p-q)^2 \Bigg( \sum_{i=1}^{n_1}x_i^2 \Bigg) \Bigg) \Bigg] = \frac{( \frac{n}{2} - \frac{(p-q)^2}{2} S^{(2)} )}{(p-q)^2}.
\end{align*}
The second equality is based on Lemma \ref{mechanism_error}. From the standard analysis on sampling process, the second term is
\begin{align*}
    \mathbb{E}\left[\left( X_{g_1} - \frac{S^{(1)}}{2} \right)^2\right] = \frac{1}{4} \mathbb{E}\left[\left( 2X_{g_1} - S^{(1)} \right)^2\right] = \frac{1}{4} S^{(2)}.
\end{align*}

Since $\mathbb{E}[M(X_{g_1})] = \mathbb{E}[X_{g_1}]$ and $\frac{S^{(1)}}{2}$ is a constant, the third term is zero. Therefore, given the above three terms, we have the error
\begin{align*}
    &\quad \mathbb{E}[(\hat{\mu}_t - \mu_t)^2]
    = \frac{\left( 2n - 2(p-q)^2 S^{(2)} \right)}{(m+n)^2(p-q)^2} + \frac{S^{(2)}}{(m+n)^2} \\
    &+\left( \frac{n_e - n}{m+n} \mu_t + \frac{(S^{(1)} - S_{e}^{(1)})}{(m+n)} \right)^2.
\end{align*}

We study the error of $\hat{\sigma}_t^2$ in SR. We denote $\sum_{i=1}^{n_2} \Phi_{2}(\Psi(x_{i, (2)}^2))$ and $\sum_{i=1}^{m_2} \Phi_{2}(\Psi(y_{i, (2)}^2))$ by $M(X_{g_2})$ and $M(Y_{g_2})$, and let $\sum_{i=1}^{n_2}x_i^2$ and $\sum_{i=1}^{m_2}y_i^2$ by $X_{g_2}$ and $Y_{g_2}$. The estimated variance (after the attack) can be written as
\begin{align*}
    \frac{2}{n + m} \left( M(X_{g_2}) + M(Y_{g_2}) \right) - \hat{\mu}_t^2 = \hat{\sigma}_t^2.
\end{align*}
Thus we have the expectation of $\hat{\sigma}_t^2$
\begin{align*}
    &\quad \mathbb{E}(\hat{\sigma}_t^2) = \frac{2}{m+n} \mathbb{E}\left[ M(X_{g_2}) + M(Y_{g_2}) \right] - \mathbb{E}[\hat{\mu}_t^2] \\
    &= \frac{m+n_e}{m+n}(\mu_t^2 + \sigma_t^2) + \frac{(S^{(2)} - S^{(2)}_e)}{m+n} - (Var(\hat{\mu}_t) + \mathbb{E}[\hat{\mu}_t]^2).
\end{align*}
We calculate the error $\mathbb{E}[(\hat{\sigma}_t^2 - \sigma_t^2)^2] = Var[\hat{\sigma}_t^2] + (\mathbb{E}(\hat{\sigma}_t^2) - \sigma_t^2)^2$.
The bias is known since we know the expectation $\mathbb{E}[\hat{\sigma}_t^2]$. Next we work on the term $Var[\hat{\sigma}_t^2]$
\begin{align*}
    Var[\hat{\sigma}_t^2] = Var\left[ \frac{2}{m+n} ( M(X_{g_2}) + M(Y_{g_2}) ) \right] + Var[\hat{\mu}_t^2].
\end{align*}
Similar to the analysis of $Var[\hat{\mu}_t]$ which is $Var[ \frac{2}{m+n} ( M(X_{g_1}) + M(Y_{g_1}) ) ]$, we denote $\sum_{i=1}^{n}x_i^4$ as $S^{(4)}$ and have
\begin{align*}
    &\quad Var\left[ \frac{2}{m+n} ( M(X_{g_2}) + M(Y_{g_2}) ) \right] 
    = \frac{4 Var\left[ M(X_{g_2}) \right]}{(m+n)^2} \\
    &= \frac{2n-2(p-q)^2 S^{(4)}}{(m+n)^2(p-q)^2} + \frac{S^{(4)}}{(m+n)^2}.
\end{align*}

For the term $Var[\hat{\mu}_t^2] = \mathbb{E}[\hat{\mu}_t^4] - \mathbb{E}[\hat{\mu}_t^2]^2$, we have $\mathbb{E}[\hat{\mu}_t^2]^2 \geq 0$ and $\hat{\mu}_t \leq 1$. Thus, it is bounded by $Var[\hat{\mu}_t^2] = \mathbb{E}[\hat{\mu}_t^4] - \mathbb{E}[\hat{\mu}_t^2]^2 \leq \mathbb{E}[\hat{\mu}_t^4] \leq 1$
Given $Var[\hat{\mu}_t^2]$, $Var\left[ \frac{2}{m+n} ( M(X_{g_2}) + M(Y_{g_2}) ) \right]$ and $\mathbb{E}(\hat{\sigma}_t^2)$, we have the upper bound of error $\mathbb{E}[(\hat{\sigma}_t^2 - \sigma_t^2)^2]$.

\section{Error of OPA on PM}
\label{appendix_proof_error_output_attack_PM}
Since the proof is the same as \textsf{OPA} against SR, we omit the details and use the same set of notations. We first analyze the error of $\hat{\mu}_t$ under the PM mechanism. 
The expectation of $\hat{\mu}_t$ is
\begin{align*}
    \mathbb{E}(\hat{\mu}_t) 
    = \frac{m+n_e}{m+n}\mu_t + \frac{1}{m+n}(S^{(1)} - X_e).
\end{align*}
Then we can calculate the error
\begin{align*}
    \mathbb{E}[(\hat{\mu}_t - \mu_t)^2] = Var[\hat{\mu}_t] + (\mathbb{E}(\hat{\mu}_t) - \mu_t)^2.
\end{align*}
The bias is known due to $\mathbb{E}(\hat{u}_t) = \mu_t + \frac{1}{m+n}(S^{(1)} - X_e)$. 
We expand the variance $Var[\hat{\mu}_t]$ to the same three terms as in the analysis of \textsf{OPA} against SR. The first term equals
\begin{align*}
    \frac{\frac{n}{2} (e^{\epsilon/2} + 3)}{3(e^{\epsilon/2} - 1)^2} + \frac{S^{(2)}}{2(e^{\epsilon/2} - 1)}.
\end{align*}
From the standard analysis on sampling process, the second term is
\begin{align*}
    \mathbb{E}\left[\left( X_{g_1} - \frac{S^{(1)}}{2} \right)^2\right] = \frac{1}{4} \mathbb{E}\left[\left( 2X_{g_1} - S^{(1)} \right)^2\right] = \frac{1}{4} S^{(2)}.
\end{align*}

Since $\mathbb{E}[M(X_{g_1})] = \mathbb{E}[X_{g_1}]$ and $\frac{S^{(1)}}{2}$ is a constant, the third term is zero. Given the above three terms, the error is
\begin{align*}
    \mathbb{E}[(\hat{\mu}_t - \mu_t)^2]
    &= \frac{2n(e^{\epsilon/2} + 3)}{3(m+n)^2(e^{\epsilon/2} - 1)^2} + \frac{(1 + e^{\epsilon/2}) S^{(2)}}{(m+n)^2 (e^{\epsilon/2} - 1)} \\
    &+ \left( \frac{n_e - n}{m+n} \mu_t + \frac{(S^{(1)} - S_{e}^{(1)})}{(m+n)} \right)^2.
\end{align*}

We next study the error of $\hat{\sigma}_t^2$ under the PM mechanism. 
Similar to the analysis of \textsf{OPA} on SR, the expectation of $\hat{\sigma}_t^2$ equals
\begin{align*}
    \frac{m+n_e}{m+n}(\mu_t^2 + \sigma_t^2) + \frac{(S^{(2)} - S_{e}^{(2)})}{m+n} - (Var(\hat{\mu}_t) + \mathbb{E}[\hat{\mu}_t]^2).
\end{align*}
Then we calculate the error
\begin{align*}
    \mathbb{E}[(\hat{\sigma}_t^2 - \sigma_t^2)^2] = Var[\hat{\sigma}_t^2] + (\mathbb{E}(\hat{\sigma}_t^2) - \sigma_t^2)^2.
\end{align*}
We know the bias as the expectation $\mathbb{E}[\hat{\sigma}_t^2]$ is known. Next we study the term $Var[\hat{\sigma}_t^2]$
\begin{align*}
    Var[\hat{\sigma}_t^2] = Var\left[ \frac{2}{m+n} ( M(X_{g_2}) + M(Y_{g_2}) ) \right] + Var[\hat{\mu}_t^2].
\end{align*}
Similar to the analysis of $Var[\hat{\mu}_t]$ which is $Var[ \frac{2}{m+n} ( M(X_{g_1}) + M(Y_{g_1}) ) ]$, we denote the $\sum_{i=1}^{n}x_i^4$ by $S^{(4)}$ and have
\begin{align*}
    &\quad Var\left[ \frac{2}{m+n} ( M(X_{g_2}) + M(Y_{g_2}) ) \right] \\
    &= \frac{2n(e^{\epsilon/2} + 3)}{3(m+n)^2(e^{\epsilon/2} - 1)^2} + \frac{(1 + e^{\epsilon/2}) S^{(4)}}{(m+n)^2 (e^{\epsilon/2} - 1)}.
\end{align*}

For the term $Var[\hat{\mu}_t^2] = \mathbb{E}[\hat{\mu}_t^4] - \mathbb{E}[\hat{\mu}_t^2]^2$, we have $\mathbb{E}[\hat{\mu}_t^2]^2 \geq 0$ and $\hat{\mu}_t \leq 1$. Thus, it is bounded by $Var[\hat{\mu}_t^2] = \mathbb{E}[\hat{\mu}_t^4] - \mathbb{E}[\hat{\mu}_t^2]^2 \leq \mathbb{E}[\hat{\mu}_t^4] \leq 1$.
Given $Var[\hat{\mu}_t^2]$, $Var\left[ \frac{2}{m+n} ( M(X_{g_2}) + M(Y_{g_2}) ) \right]$ and $\mathbb{E}(\hat{\sigma}_t^2)$, we have the upper bound of error $\mathbb{E}[(\hat{\sigma}_t^2 - \sigma_t^2)^2]$.

To show and compare the error of \textsf{IPA} and \textsf{OPA}, we replace some terms with intermediate notations shown in Table~\ref{tab:intermedias}. The comparison results are shown in Table~\ref{tab:attack_error}. 

\begin{table}[!t]
\centering
\caption{Intermedia for Error Analysis.}
\vspace{-8pt}
\scalebox{0.85}{
\begin{tabular}{|c|l|}
\hline
\textbf{Intermedia} & \textbf{Values} \\ \hline
$\mathcal{P}$        & $\left( \frac{n_e - n}{m+n} \mu_t + \frac{(S^{(1)} - S_{e}^{(1)})}{(m+n)} \right)^2$          \\ \hline
$\mathcal{Q}$        & $\frac{(n_e+m)(\mu_t^2 + \sigma_t^2)}{(m+n)^2} + \frac{S^{(2)} - S_{e}^{(2)}}{(m+n)^2}$  \\ \hline
$S^{(p)}$        & $\sum_{i=1}^{n}x_i^p$           \\ \hline
$\eta$        & $\frac{m+n_e}{m+n}(\mu_t^2 + \sigma_t^2) - \sigma_t^2$          \\ \hline
$\mathcal{T}^{IPA}_{\mathrm{SR}}$        &  $\left( \eta + \frac{1}{m+n}(S^{(2)} - S_{e}^{(2)}) - Var^{IPA}_{\mathrm{SR}}[\hat{\mu}_t] - \mathbb{E}[\hat{\mu}_t]^2 \right)^2$         \\ \hline
$\mathcal{T}^{IPA}_{\mathrm{PM}}$        &  $\left( \eta + \frac{1}{m+n}(S^{(2)} - S_{e}^{(2)}) - Var^{IPA}_{\mathrm{PM}}[\hat{\mu}_t] - \mathbb{E}[\hat{\mu}_t]^2 \right)^2$           \\ \hline
$\mathcal{T}^{OPA}_{\mathrm{SR}}$  &  $\left( \eta + \frac{1}{m+n}(S^{(2)} - S_{e}^{(2)}) - Var^{OPA}_{\mathrm{SR}}[\hat{\mu}_t] - \mathbb{E}[\hat{\mu}_t]^2 \right)^2$ \\ \hline 

$\mathcal{T}^{OPA}_{\mathrm{PM}}$ &  $\left( \eta + \frac{1}{m+n}(S^{(2)} - S_{e}^{(2)}) - Var^{OPA}_{\mathrm{PM}}[\hat{\mu}_t] - \mathbb{E}[\hat{\mu}_t]^2 \right)^2$ \\ \hline
$Var_{\mathrm{SR}}^{IPA}[\hat{\mu}_t]$  & $\frac{2}{(m+n)(p-q)^2} - \mathcal{Q}$ \\ \hline
$Var_{\mathrm{PM}}^{IPA}[\hat{\mu}_t]$  & $\frac{2(e^{\epsilon/2} + 3)}{3(n+m)(e^{\epsilon/2} - 1)^2} + \frac{2 \mathcal{Q}}{(e^{\epsilon/2} - 1)} + \mathcal{Q}$ \\ \hline

$Var_{\mathrm{SR}}^{\textsf{OPA}}[\hat{\mu}_t]$   &   $\frac{\left( 2n - 2(p-q)^2 S^{(2)} \right)}{(m+n)^2(p-q)^2} + \frac{S^{(2)}}{(m+n)^2}$ \\ \hline
$Var_{\mathrm{PM}}^{\textsf{OPA}}[\hat{\mu}_t]$   &   $\frac{2n(e^{\epsilon/2} + 3)}{3(m+n)^2(e^{\epsilon/2} - 1)^2} + \frac{(1 + e^{\epsilon/2}) S^{(2)}}{(m+n)^2 (e^{\epsilon/2} - 1)}$ \\ \hline
$\mathbb{E}[\hat{\mu}_t]$ & $\frac{m+n_e}{m+n}\mu_t + \frac{1}{m+n}(S^{(1)} - S_{e}^{(1)})$ \\ \hline
\end{tabular}
}
\label{tab:intermedias}
\vspace{-10pt}
\end{table}

\section{Proof of Theorem \ref{OPA_better}}
\label{appendix_proof_OPA_better}
\begin{proof}
    First we study \textsf{OPA} and \textsf{IPA} in the SR mechanism. Given the error analysis of SR and PM, we have $\mathrm{Err}_{IPA}(\hat{\mu}_t) - \mathrm{Err}_{OPA}(\hat{\mu}_t) = \frac{2m}{(m+n)^2(p-q)^2} - \frac{(n_e+m)(\mu_t^2 + \sigma_t^2) - S_{e}^{(2)}}{(m+n)^2}$.
    Since $(n_e+m)(\mu_t^2 + \sigma_t^2) - S_{e}^{(2)} = \sum_{i=1}^{m} y_i^2 \leq m$, we have $\mathrm{Err}_{IPA}(\hat{\mu}_t) - \mathrm{Err}_{OPA}(\hat{\mu}_t) \geq \frac{2m - m(p-q)^2}{(m+n)^2(p-q)^2} \geq 0$.
    
    For the PM mechanism, since $\sum_{i=1}^{m} y_i^2 \geq 0$, we have
    \begin{align*}
        & \quad \mathrm{Err}_{IPA}(\hat{\mu}_t) - \mathrm{Err}_{OPA}(\hat{\mu}_t) = \frac{2m(e^{\epsilon/2} + 3)}{3(m+n)^2(e^{\epsilon/2})} + \\
        &\frac{2(\sum_{i=1}^{m}y_i^2 + S^{(2)})}{(m+n)^2(e^{\epsilon/2} - 1)} + \frac{\sum_{i=1}^{m}y_i^2 + S^{(2)}}{(m+n)^2} - \frac{(1+e^{\epsilon/2}) S^{(2)}}{(m+n)^2(e^{\epsilon/2} - 1)} \geq 0
    \end{align*}

    Then we compare the error on variance. Since $Var^{\textsf{OPA}}_{\mathrm{SR}} \leq Var^{\textsf{IPA}}_{\mathrm{SR}}$, we have $\mathcal{T}^{OPA}_{\mathrm{SR}} \leq \mathcal{T}^{IPA}_{\mathrm{SR}}$. Besides, we have $\frac{2n-2(p-q)^2 S^{(4)}}{(m+n)^2(p-q)^2} + \frac{S^{(4)}}{(m+n)^2} \leq \frac{2}{(m+n)(p-q)^2} - \frac{S^{(4)}}{(m+n)^2}$, then it turns out the upper bound of $\mathrm{Err}_{OPA}(\hat{\sigma}^2_t) \leq \mathrm{Err}_{IPA}(\hat{\sigma}^2_t)$ on SR. By the similar calculation, we have the same conclusion on PM mechanism.
    
\end{proof}

\section{Proof of Theorem \ref{consistency_OPA}}
\label{appendix_proof_consistency_OPA}
\begin{proof}
    According to the attack error, we calculate the derivative of attack error on mean and the upper bound of the attack error on variance, and have all derivatives negative for all $\epsilon > 0$. In other words, the attack error on mean and the upper bound of attack error on variance decrease as $\epsilon$ grows.
\end{proof}

\end{appendix}
\end{document}